\documentclass[]{amsart}
\usepackage{latexsym,amssymb,amsmath,amsthm}
\usepackage{comment}
\usepackage{float}
\usepackage{graphicx}
\newtheorem{thm}{Theorem}
\newtheorem{coro}[thm]{Corollary}
\newtheorem{lemma}[thm]{Lemma}
\newtheorem{propo}[thm]{Proposition}
\theoremstyle{definition}

\newcommand{\R}{\mathbb{R}}
\newcommand{\Z}{\mathbb{Z}}

\def\G{\mathcal{G}}

\def\tr{\rm{tr}}
\def\str{\rm{str}}

\title{Isospectral deformations of the Dirac operator}
\author{Oliver Knill}
\date{June 22, 2013}
\address{
        Department of Mathematics \\
        Harvard University \\
        Cambridge, MA, 02138
        }

\begin{document}
\maketitle
\begin{abstract}
We give more details about an integrable system \cite{diracannouncement} in which the Dirac 
operator $D=d+d^*$ on a graph $G$ or manifold $M$ is deformed using a Hamiltonian system 
$D'=[B,h(D)]$ with $B=d-d^* + \beta i b$.
The deformed operator $D(t) = d(t) + b(t) + d(t)^*$ defines a new exterior derivative $d(t)$ and 
a new Dirac operator $C(t) = d(t) + d(t)^*$ and Laplacian $M(t) = C(t)^2$ and so a new distance 
on $G$ or a new metric on $M$. For $\beta=0$, the operator $D(t)$ stays real for all $t$. 
While $L=M(t)+V(t)$ does not change, the new Laplacian $M(t)=C(t)^2$ and the emerging potential 
$V(t)=b^2$ do evolve. The operators $M,V$ are always real and commute. 
The cohomology defined by the deformed exterior derivative $d(t)$ is the same as for $d=d(0)$ 
as we can explicitly deform cocycles and coboundaries. 
The new Dirac operator $C(t)$ defines a new metric, so that the isospectral flow is an 
evolution which deforms the geometry as defined by zero forms. If $U'=BU$ is the associated unitary, 
then the McKean-Singer formula $\str(U(t))=\chi(G)$ still holds. 
While super partners $f,D(t)f$ span the same plane at all times, observable super symmetry fades: if 
$f$ is an eigenvector of $L$ which is a fermion - an eigen-$2k+1$-form of $L$ for some integer $k$ - 
then $D(t)f$ is only bosonic at $t=0$ and the 
angle between the fermionic subspace and $D(t) f$ goes to zero exponentially fast. The coordinate system has
changed so that the original superpartner $D(0) f$ is now far away for the new geometry. 
The linear relativistic wave equation $u(t)''= -L u(t)$ and its solution 
$u(t)=\cos(Dt) u(0) + \sin(D(t) t) D^{-1} u'(0) = e^{i D} (u(0) - i D^{-1} u'(0)$
with fixed $D$ is not affected by the symmetry since only $L$ and not $D$ enters in the solution formula.
But the preperation of the initial velocity, the 
nonlinear solution $u(t) = \cos(D(t) t) + \sin(D(t) t) D(0)^{-1} u'(0)$ of
the wave equation with time dependent $D$ or the unitary evolution $U(t)$ defined by the deformation depends on $D$.
The evolution has a geometric effect: 
with $d(t)$ as a new exterior derivative, the property $d(t) \to 0$ for $|t| \to \infty$ implies that 
space expands, with a fast inflationary start. The inflation rate can be tuned by scaling $D$.
Instead of solutions to the wave equation, the nonlinear evolution has more
soliton -  and particle - like solutions which feature interaction with adjacent forms. 
In the limit $t \to \pm \infty$, the operator $D$ becomes block diagonal 
$b_+ = -b_-$ with $b_{\pm}^2=L$, leading to linear solutions $\cos(bt) u(0) + \sin(bt) b^{-1} u'(0)$ of the wave 
equation which leaves each space $\Omega_p$ of $p$-forms invariant. 
For $G=K_2$, explicit formulas illustrate the inflation. We also look at the circle case, where
already the 0-form and 1-form spaces can be joint by solutions of the wave equation.
The nonlinear Dirac wave equation uses the entire geometric space but
asymptotically, we get linear wave equation which preserve each p-form subspace and which is relativistic 
quantum mechanics and classical Riemannian geometry. Geometry alone can lead to interesting
nonlinear wave dynamics, the emergence of new dimensions, complex structures, inflation and to a geometric
toy model featuring Riemannian or graph geometry in which super symmetry is present but 
difficult to measure.
\end{abstract}

\section{Summary} 

A brief version of these notes was a condensation of this document and 
posted as \cite{diracannouncement}. This current document is actually our first
writeup about this integrable system and aims to give more details without trying to be brief. We start with 
an extended abstract. The next section gives the background definitions. Then we prove the
results and set up everything in the simplest possible case, where distances matter: the circle in the 
manifold case and the two point graph $K_2$ in the discrete case. In an appendix, we place the system into 
the context of other known integrable systems. \\

By deforming the Dirac operator $D = d + d^*$ on a finite simple graph $G=(V,E)$ or Riemannian 
manifold $(M,g)$, using integrable Hamiltonian systems $D'=J \nabla H(D) = [B,h'(D)]$ with 
$B=d-d^* + \beta i b, D(t)=U(t)^* D U(t) = d(t) + d(t)^* + b(t)$, we alter the exterior derivative $d$, the distance in $G$ or 
$M$ and obtain a natural symmetry for the quantum mechanical system defined by $D$ on the graph $G$ or manifold $M$. 
This nonlinear evolution $U(t)$ can be used as a nonlinear alternative of the wave equation, if $\beta$ is nonzero. 
Unlike the linear Dirac wave equation, which is not affected by the nonlinear symmetry except for preparing the initial velocity 
of the wave, the nonlinear flow influences the geometry. If we allow the flow to become complex, that is with $\beta \neq 0$, then
the nonlinear flow is asymptotically indistinguishable from the linear wave equation.  Mathematically, one can see this as follows: 
the wave equation is linear in a static space and even so we use the Dirac operator $D$ to describe its
solutions, only the Laplacian $L=D^2$ matters in the linear case; one can see this by making a 
Taylor expansion of the explicit d'Alembert solution formulas which involve $D$. 
But this situation changes with the nonlinear evolution: the Dirac operator $D$ enters now 
the stage, despite the fact that many things remain the same: the spectrum of $D(t)$, the operator $L=D(t)^2$ 
itself as well as the cohomology defined by the exterior derivative $d(t)=D(t)^+$ are not affected by the deformation.
Each of the traces $H(D) = \tr(D^k)$ defines an integral of motion and can be used as Hamiltonians of 
the system $\dot{D} = J \nabla H(D)$. They are the Noether invariants of the group actions. In the manifold case,
$\tr(D^k)$ is defined in a zeta-regularized form 
as $\zeta(-k)$ where $\zeta(s)$ is the Dirac zeta function of the manifold. Unlike the zeta function
of the Laplacian, the Dirac zeta function can be analytically continued to the entire plane, at least for odd-dimensional
manifolds (for even dimensional manifolds, some poles of the Laplace zeta function can remain). 
While classical quantum mechanics, expressed as the classical Schr\"odinger equation $u' = i L$ 
on space does not change, since $L=L(t)$ stays constant, relativistic
quantum mechanics given by the Dirac equation $u'=\pm i D(t) u$ or the new nonlinear equation $u'=B(t) u$ does evolve 
now with a time-dependent Dirac operator. The deformation of the geometry has interesting effects, despite the fact
that the operator $L$ does not move. The Laplacian $L$ becomes the sum of a shrinking Laplacian 
$M(t) = (d(t) + d(t)^*)^2 = C(t)^2$ with respect to the new exterior derivative $d(t)$ and an expanding potential 
$V(t) = b(t)^2$. These two ingredients of the Laplacian are real and commute for all $\beta$ and all $t$. 
Distances in the graph defined by the new Dirac operator $C(t)$ increases in all parameter directions 
when starting with the standard Dirac operator $D(0) = d+d^*$ given by the standard exterior derivative $d$. 
Independent of the direction in which we move, there is an ``arrow of time" given by the expansion: the 
symmetry can actually be used as ``time" and it really does not matter which direction is taken on the symmetry 
group starting from $t=0$, the features of the time evolution are always the same.
We verify that the McKean-Singer formula ${\rm Re}(\str(U(t))) = \chi(G)$ 
holds for the nonlinear deformation in the graph case. But if $f$ is an eigenfunction of $L$ which is a fermion, then
its new super partner $D(t) f$ of a fermion $f$ is no more a boson. 
While mathematical super symmetry \cite{Witten1982,Cycon} and McKean-Singer symmetry is still present, the 
super partner $D(t) f$ is an interaction state close to the fermionic subspace. Assume we look at a fixed vector
$f$ which is a fermion.  The vector $D(t) f$ is already after a short time indistinguishable from being a fermion. 
The evolution has pushed it from the boson space $\Omega_b$ close to the fermion space $\Omega_f$. 
While by the unitary nature of the evolution, $f(t)$ and $g(t)$ stay perpendicular at all times, the operators $f,D(t) f$ do not. 
Finally, for any graph with at least one edge, the expansion of the Connes pseudo metric defined by 
$C(t) = d(t) + d(t)^*$ or the Riemannian metric defined directly by the new Laplacian 
$M(t)=C(t)^2= d(t) d(t)^* + d(t)^* d(t)$ 
features a fast inflationary initial growth, which then slows down exponentially. 
Tuning of the coupling constants of the exterior derivative $d: \Omega_p \to \Omega_{p+1}$ by changing the 
Dirac operator to $D = \sum_p \gamma_p (d_p + d_p^*)$ corresponds to choosing units on each $\Omega_p$ ``brane" 
and can lead to an evolution, where scales of hierarchies can emerge: the physics on the different p-forms can be
drastically different after some time. The scaling factors do not change the symmetries of the spectrum since
$\{D,P\}$ is still zero and because super symmetry still holds for $L=D^2$. 
These features are independent of the Hamiltonian $H$ and of the time direction. 
The only input is geometry, the initial graph or manifold. The nature of the Hamiltonian system assures that 
the eigenvectors $f,g(t)=D(t) f$ of $L$ are only perpendicular when $t=0$, 
if $f$ is kept to be a fermion. Of course $f(t)=U(t) f,g(t) = U(t) g$ would stay perpendicular and $f,D(0) f$
are always perpendicular. But $f,D(t) f$ will be more correlated now and $D(0) f$ is now far way from $f$ when using
the wave equation with $D(t)$. It is the Dirac operator $D$ or equivalently the Dirac wave evolution $e^{i D t}$ 
which measures distances. Now $D(t)$ changes and instead of $e^{i D t}$ or $e^{i D(t) t}$ we propose to look at
$U(t)$. Since $D(t)$ converges to $D(\infty)$, this is for larger $t$ very close to the wave solution 
$e^{i D(\infty) t}$. It plays therefore also the role of the exponential map and parallel transport. \\

Since the Hamiltonian system is a geometric evolution equation which alters geometry, it
has the potential for the study the topology of graphs or manifolds. 
It is still unknown what the deformation means geometrically, when restricted to $k$-forms.
Both in the Riemannian or graph case, we would like to know the evolution of curvature $K(x)$. 
Curvature for graphs satisfying Gauss-Bonnet-Chern can be extended to the manifold traced out by 
the wave equation. Having $K(x)$ represented as the expectation of the index 
$K(x) = {\rm E}[i_f(x)]$ on a probability space $\Omega$ of functions $f$ on the space, the dynamics can 
be used to deform curvature as the push forward of the probability measure on $\Omega$ under the evolution. We can therefore
define curvature also for the deformed geometries and in the graph case extend it to the continuum defined by the 
dynamics. While the linear wave equation does not change curvature, 
the nonlinear flow will. The probabilistic representation of curvature as an average of index is also crucial to 
define curvature on each of the linear spaces defined by $k$-forms as well as their completions by the wave equation
on which $U(t)$ replaces the exponential map and parallel transport. \\

Here is an other attempt to summarize the core mathematics:  \\

The Dirac deformation $\dot{D}=[B,D]$ is completely integrable. On each invariant two-dimensional McKean Singer
plane spanned by $f,D(t)f$ with an eigenfunction $f$ of $L$, we can describe the motion of $D(t) f$. 
Each $D(t)=d(t)+d(t)^*+b(t)$ defines a new exterior derivative $d(t)$ for which the cohomology is the same than for $d_0$. 
We can write $L=D(t)^2= M(t) + V(t)$ as a sum of two commuting operators $M(t)=(d(t)+d(t)^*)^2 \to 0$ and $V(t)=b(t)^2$. 
If $D(t)=U(t)^* D(0) U(t)$, then $\str(U(t))=\str(1)$ is the Euler characteristic for all $t$. For nonzero $\beta$, 
the nonlinear evolution becomes asymptotic to a linear Dirac transport equation $\dot{u} = i D_{\infty} u$ which together
with $\dot{u}=-i D_{\infty} u$ build solutions to the wave equation $\ddot{u} = - Lu$. 
While the nonlinear Dirac deformation does not influence the Schr\"odinger flow nor the classical linear wave evolutions, 
it is invisible for classical linear wave evolution on the graph or Riemannian manifold, it becomes relevant
if the Dirac evolution $e^{i Dt}$ is replaced by the new nonlinear evolution $U(t)$.  \\

\begin{figure}
\scalebox{0.45}{\includegraphics{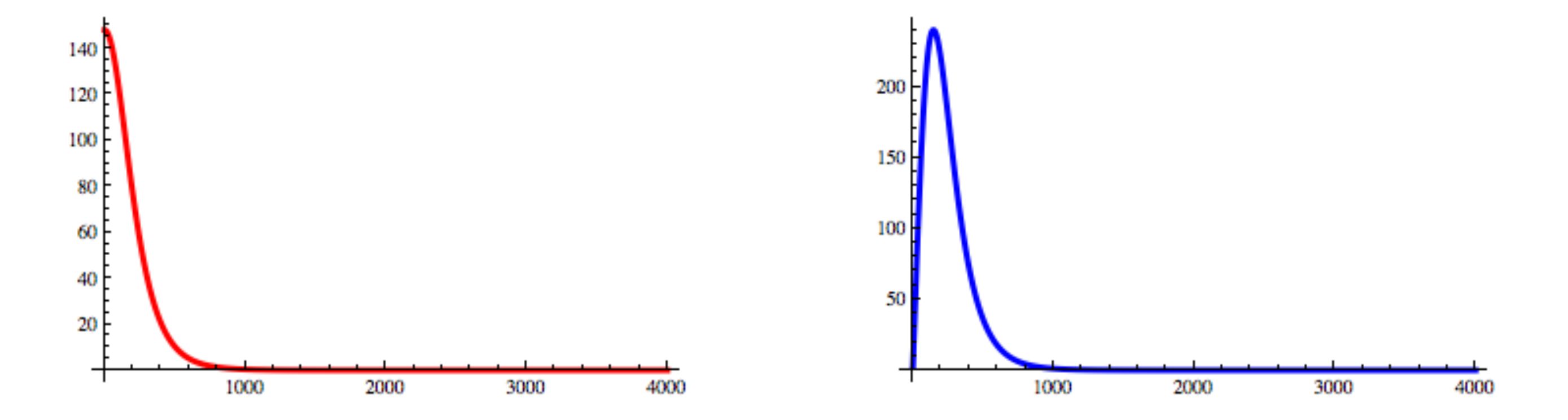}}
\caption{
The figure to the left shows $\tr(M(t))$, integrated numerically with Runge-Kutta from $t=0$ to $t=4$. We prove here that 
this quantity is a Lyapunov function: it monotonically decreases with $t$: the positive operator $M$ satisfies 
${\tr}(M(t)) \to 0$ for $|t| \to \infty$. 
The second graph shows the graph of $-\frac{d}{dt} \tr(M(t))$, which we prove to be always positive, zero
at $t=0$ and for $|t| \to \infty$. The nature of the differential equations make it look to be of logistic 
nature. The fact that energy conservation $L=M(t) + V(t)$ is constant will force it to have an
inflation bump in the graph of $\tr(M(t))$ and the solutions then slows down exponentially. The later means 
that the expansion of space is asymptotically proportional to the diameter of the expanding space. It is 
very early in the evolution that nonlinear effects are important. 
The strong expansion has a focusing influence on the otherwise diffusive nature of solutions to
the wave equation. Soliton-like and particle-like solutions are more likely to form. 
\label{inflation}
}
\end{figure}

\begin{figure}
\scalebox{0.12}{\includegraphics{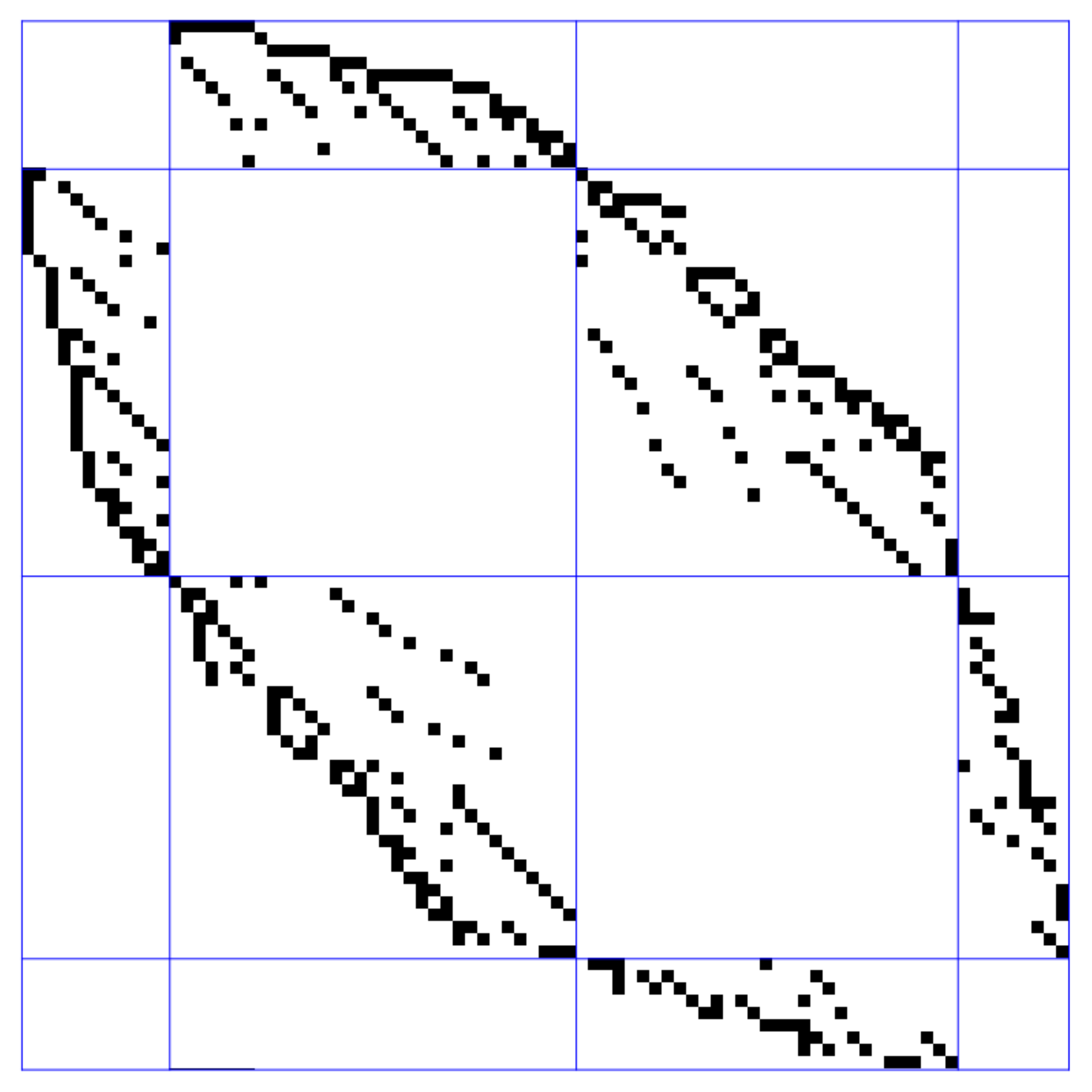}}
\scalebox{0.12}{\includegraphics{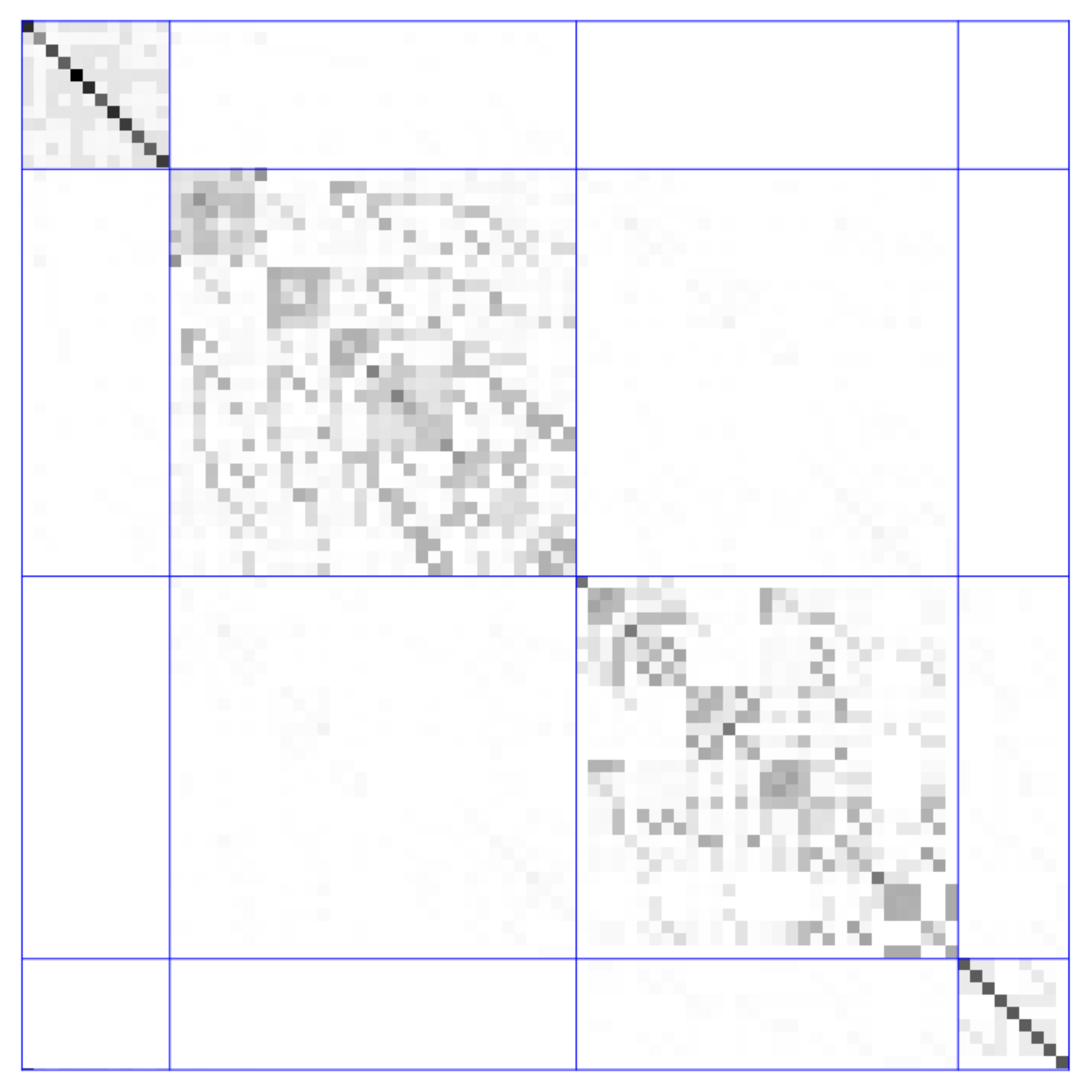}}
\scalebox{0.12}{\includegraphics{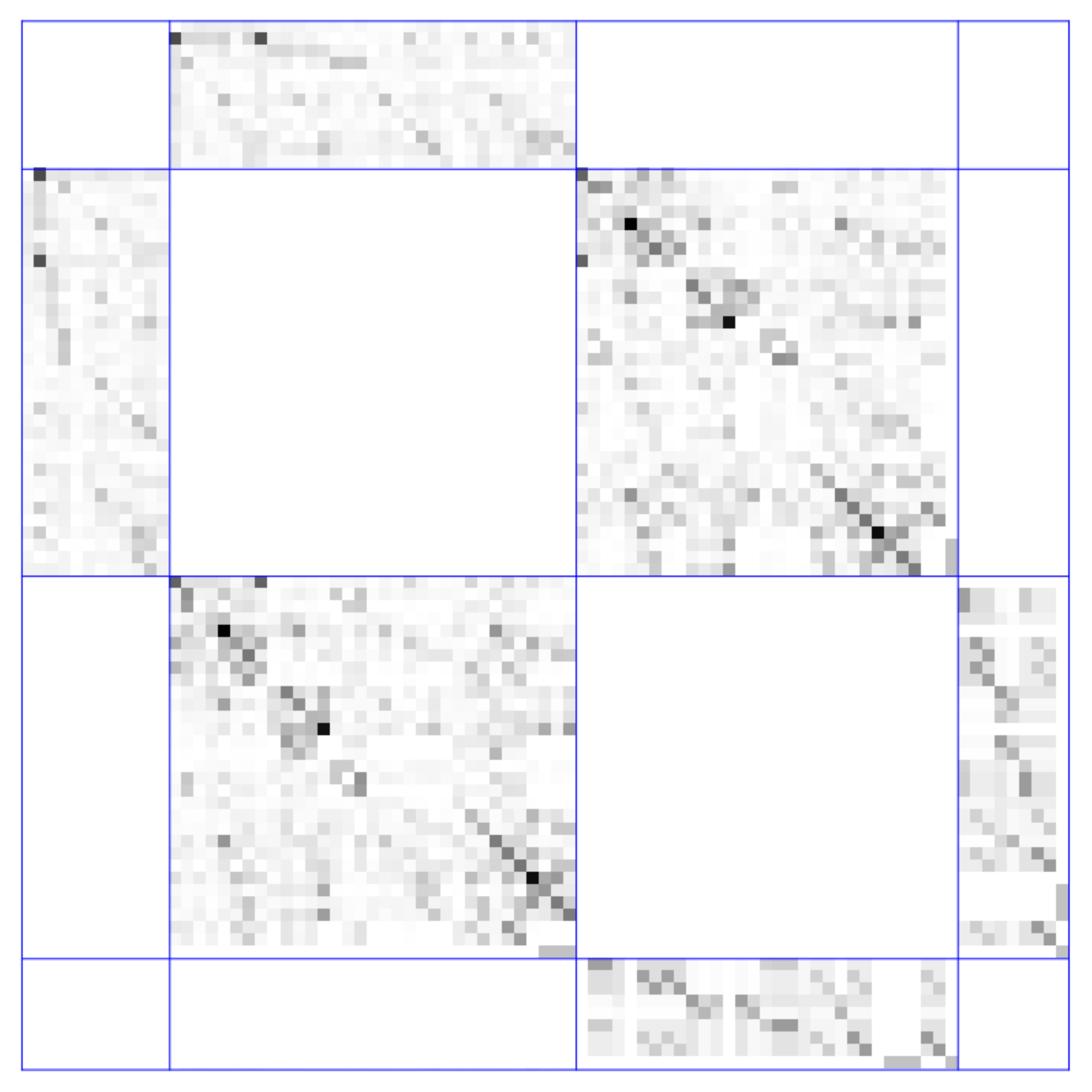}}
\caption{
The left figure shows the operator $D(0)$. In the middle, we see the 
deformed operator $D(1)=C(1)+b(1)=d(1)+d(1)^*+b(1)$. At time $t=1$ already, the evolution
has passed the inflationary expansion and has moved already pretty close to its final shape $b(\infty)$. 
The $C(1)$ part is so small already that it can not be seen in the middle picture. 
To the right, we see $C(1)= d(1) + d(1)^*$ which had been rescaled to become visible. 
While $d(1)$ is small, it is not zero. Still, $d(1)$ can be used as a deformed exterior 
derivative with the same cohomology than $d(0)$.}
\end{figure}

\begin{figure}
\scalebox{0.12}{\includegraphics{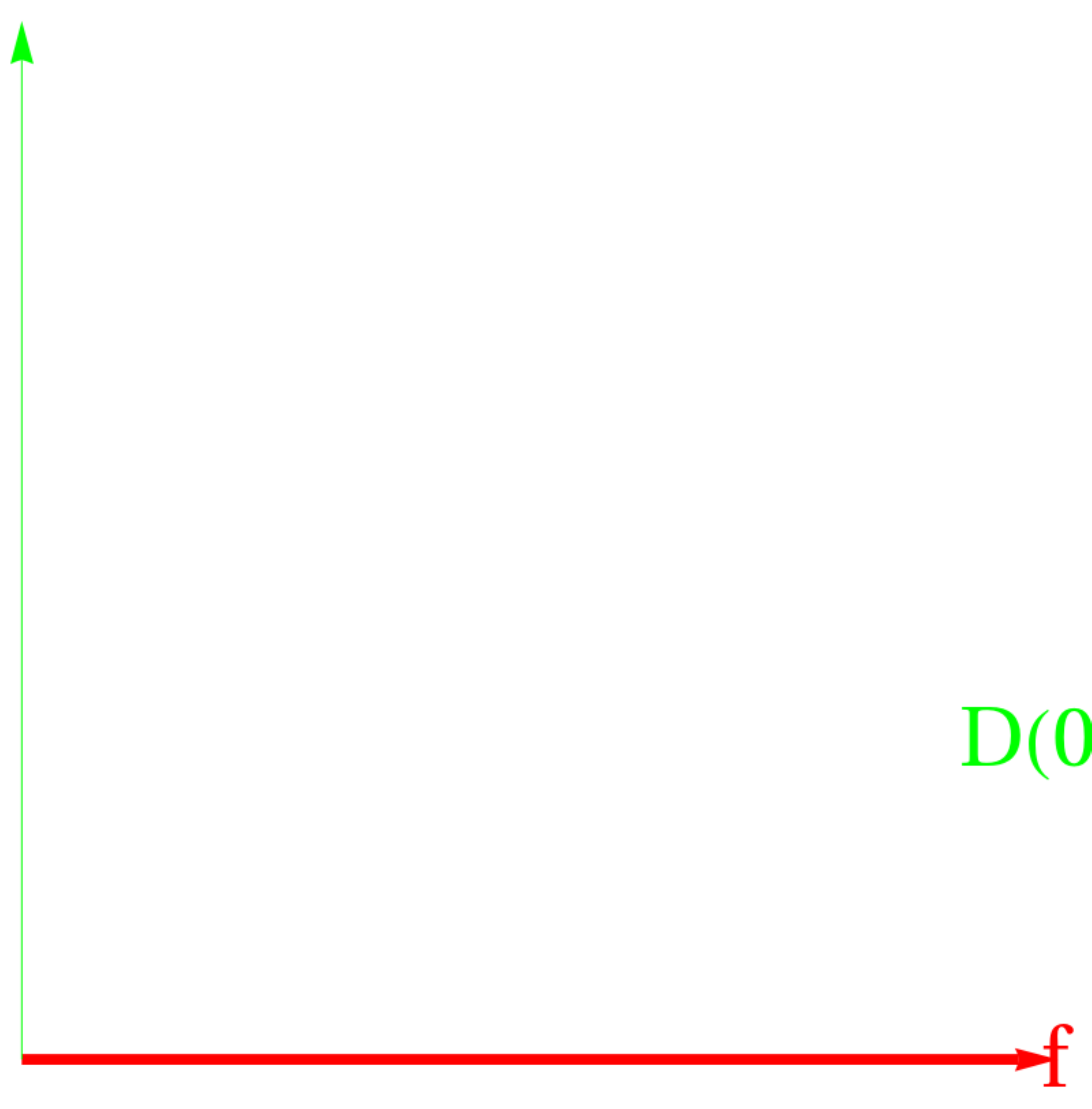}}
\scalebox{0.12}{\includegraphics{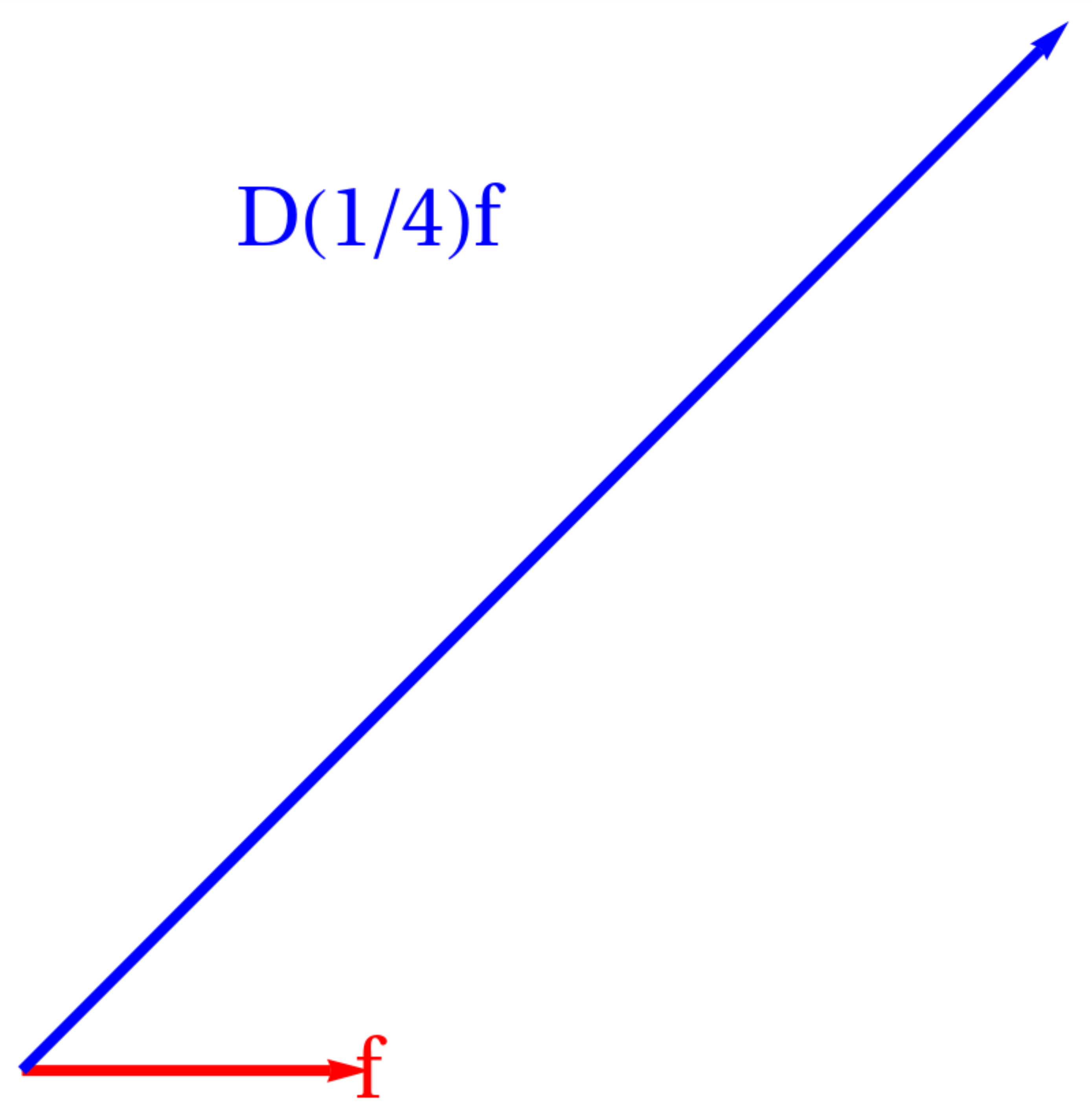}}
\scalebox{0.12}{\includegraphics{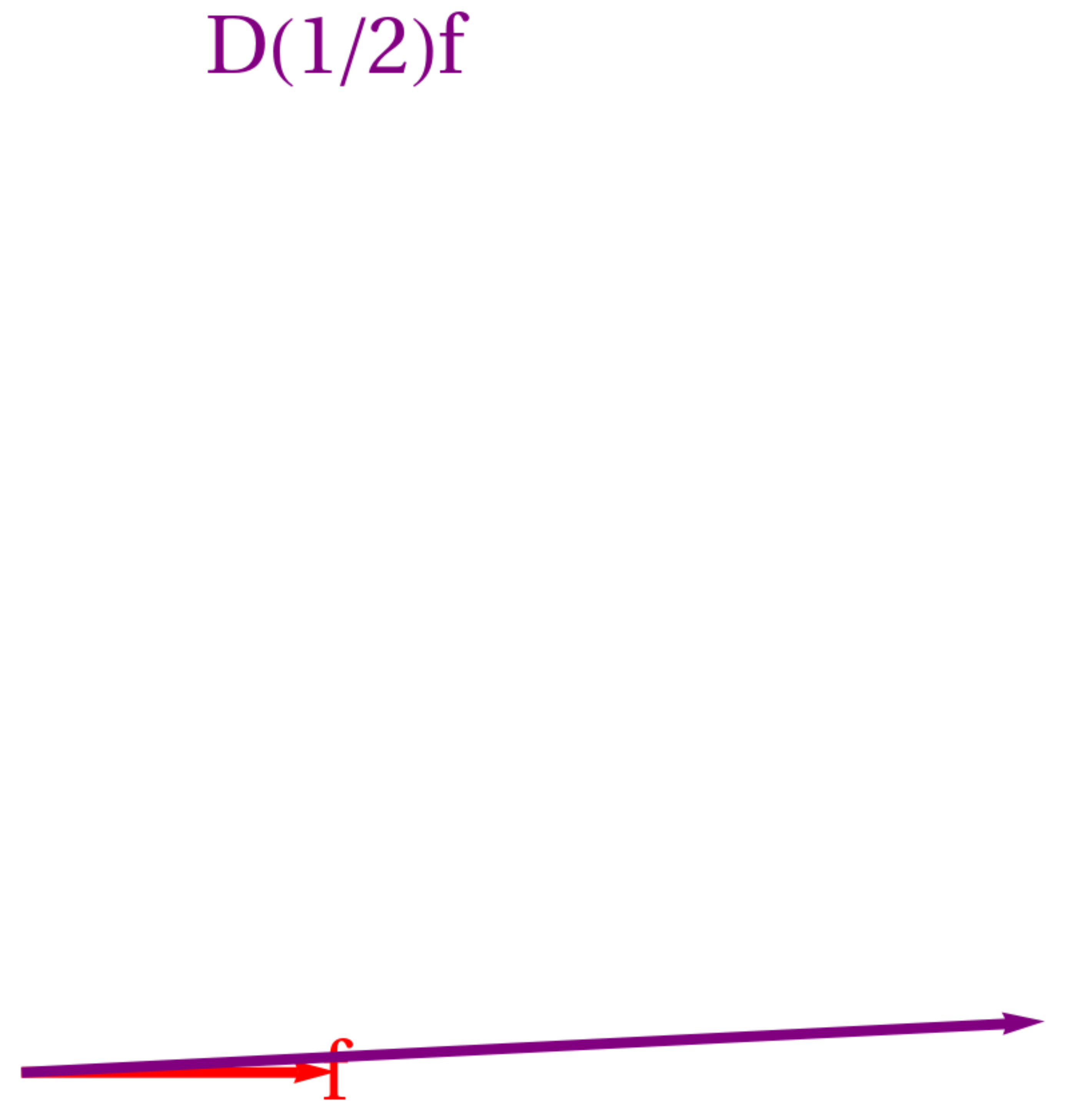}}
\caption{
In this figure, we see two vectors $f,D(t)f$ which span a McKean-Singer plane. 
The eigenvector $f$ of $L$ and does not change. Also $D(t) f$ is an eigenvector. 
At $t=0$, the two vectors $f,D(t) f$ are perpendicular as would be 
$f(t)=U(t) f, g(t) = U(t) g$. For $t=1$ already,
the vector $D(1)f$ is strongly correlated to $f$. If $f(0)$ is a pure fermion
then $g(t)$ is a pure boson only for $t=0$. We can not see the super partner
already after a relatively short time because the angle between the fermion subspace
and $g(t)$ has become exponentially close. Super-symmetry - that is a pairing between bosonic and fermionic
eigenvectors given by the McKean Singer theorem - is still present but seen only at 
the very beginning because the deformation will change the way how waves move, the 
super partner $D(0)f$ will appear far away from $f$. 
Eigenfunctions to the same energy are do not need to be perpendicular, even if the matrix is self-adjoint. }
\end{figure}

\begin{figure}
\scalebox{0.12}{\includegraphics{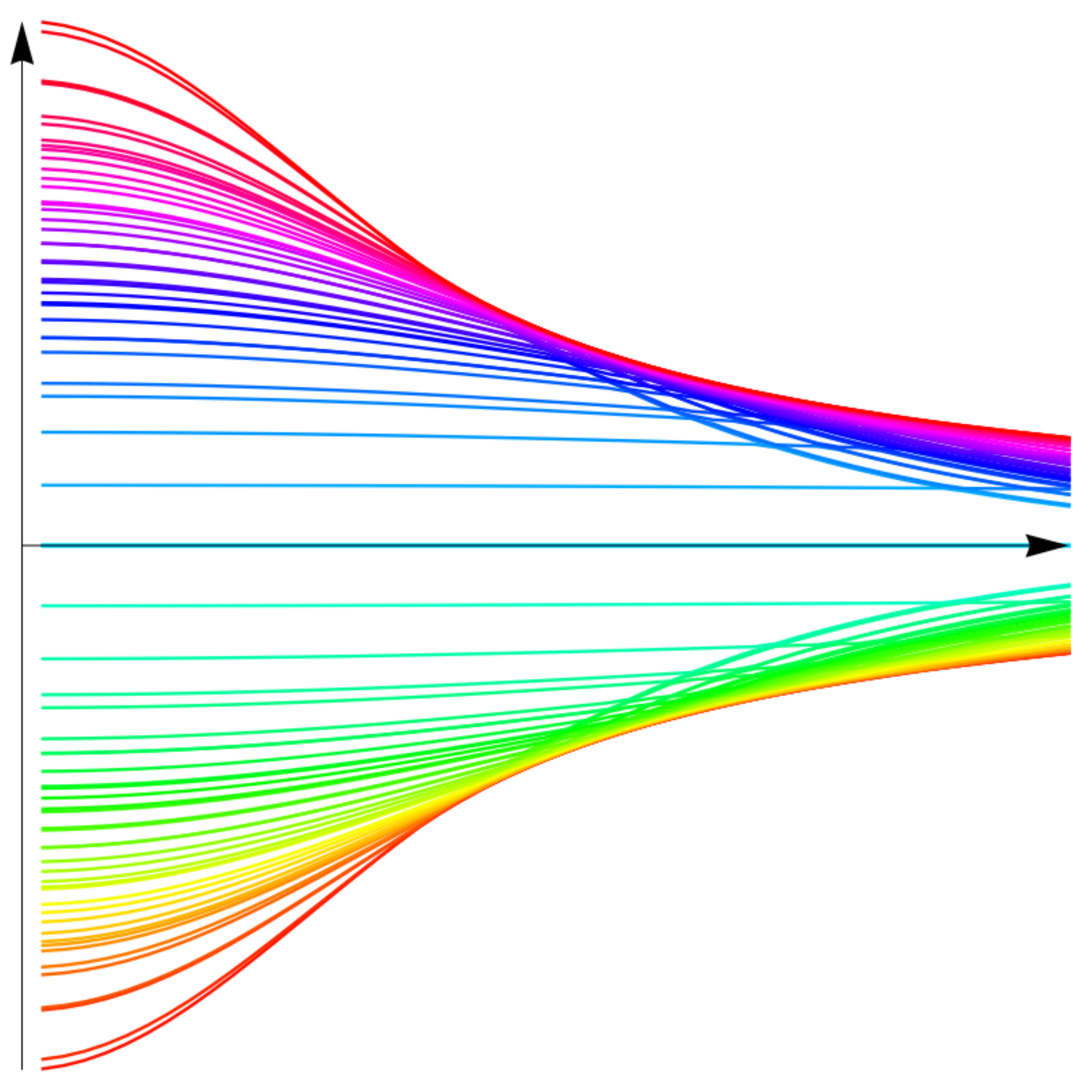}}
\scalebox{0.12}{\includegraphics{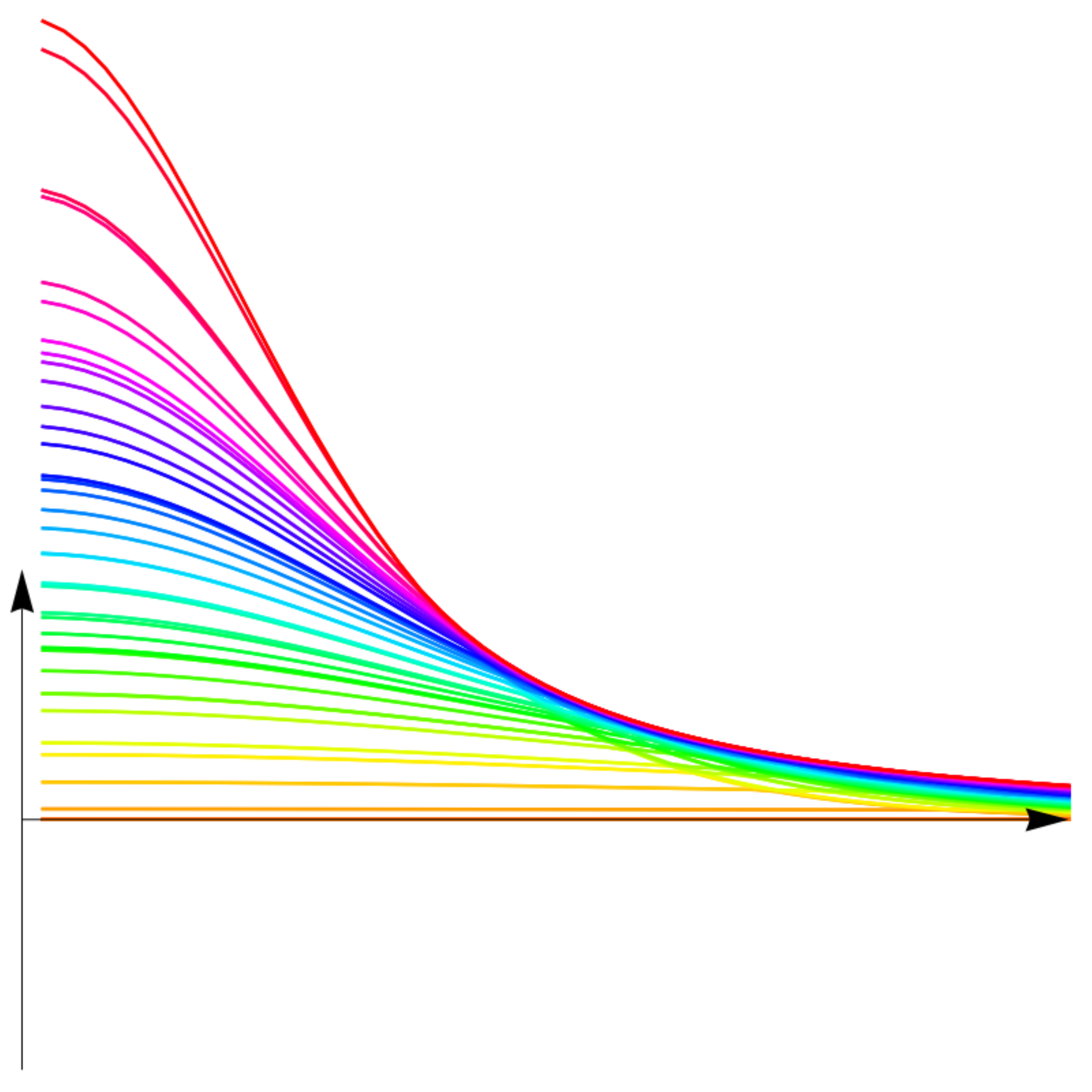}}
\scalebox{0.12}{\includegraphics{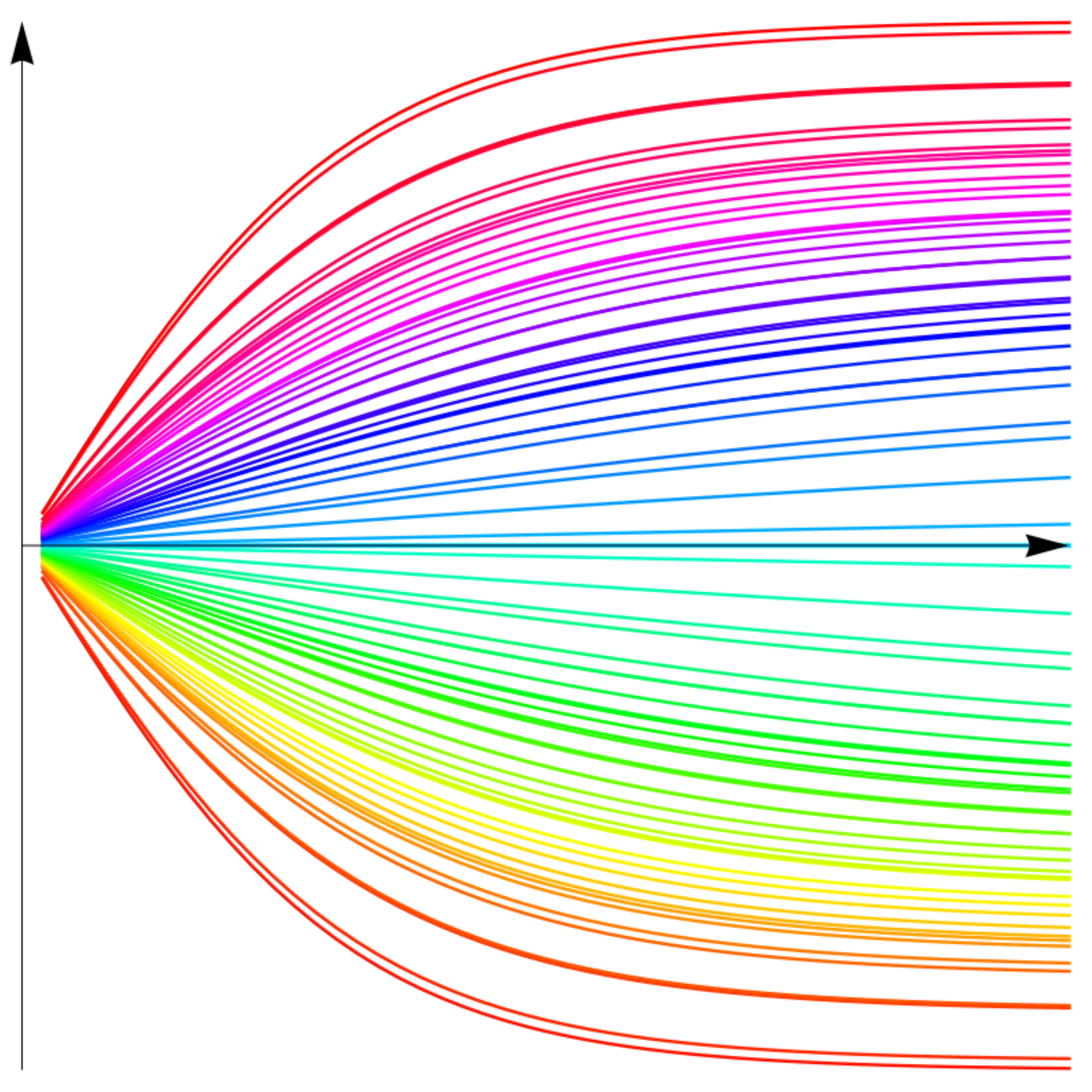}}
\caption{
We see the simultaneous motion of the spectrum of the matrices $C(t),M(t)$ and $b(t)$
for a random graph. As the example displays, it can happen that different eigenvalues of the 
geometric Dirac operator $C(t)$ cross over. A study of the spectral motion of $M(t)$ has
not yet been done. It would allow us to deduce information about the geometry. 
The last figure shows the eigenvalues of the block diagonal ``dark energy" 
operator $b(t)$.}
\end{figure}

Lets add some more informal remarks, which are irrelevant for the rest. 
What does it mean that the super partner of $f$ are difficult to be 'seen' after some time? 
Assume we have a fermion $f \in \Omega_f$, then $g=D(0) f$ is a boson. Now lets evolve time
to get an operator $D(t)$. This means we have lost $D(0)$ which allowed us to go from $f$ to $g=D(0) f$. 
At time $t$, it is $D(t)$ or its geometric part $C(t)$ which governs geometry at time $t$ determined 
by the wave equation. The deformed operators have changed the
geometry and the partner $g$ is now expected further away. We can measure distances
in the computer, it just involves linear algebra to measure the time a wave takes
to get from $f$ to $g$. If $f(t),g(t)$ were evolved together, then of course,
they stay perpendicular, but now both are neither fermion nor boson. The coordinate transformation
which brings $f(t)$ to the fermion space is not the same than the coordinate transformation
which make $g(t)$ a boson. Unlike at $t=0$, we have a coordinate system in which the
super partners are present but not {\bf close enough} with the wave equation
we have at hand at time $t$. In other words, moving in the symmetry group of the geometric space
has changed distances in such a way that initially close super partners are now more remote. \\

When working on a graph, we just need linear algebra and ordinary differential equations \cite{DiracKnill}.
All this has been implemented on the computer already and the results of this paper were
mostly discovered first experimentally. The operator $D$ constructed from a finite graph 
is a finite matrix which the symmetry deforms with time. 
We can now make measurements at time $t$ by ``looking around" in the space by 
``sending light" in different directions, using $D(t)$.  Sending a wave from one vertex to an other
just needs to solve a linear system of equations to find the right initial velocity. The length 
of the initial velocity is inversely proportional to the distance. 
Also $k$-forms evolve by the wave equation triggered by the operator $D(t)$. 
\vspace{1mm}
\parbox{16.8cm}{
\parbox{5.5cm}{
It is now  awfully tempting to associate parts of
$D$ of a 3-dimensional
graph or manifold with physics:
}
\parbox{10.5cm}{
\begin{tiny}
\begin{tabular}{|c|c|c|c|cc}
  boson   &  fermion &  boson   & fermion  &           &        \\ \hline
  graviton& electron &          &          &   gravity & U(0)   \\ \hline
  higgs   & lepton   &  photon  &          &   electro & U(1)   \\ \hline
          & neutrino &  B,B,Z   &  quark   &   weak    & SU(2)  \\ \hline
          &          &  gluon   & hadron   &   strong  & SU(3)  \\ \hline
\end{tabular}
\end{tiny}
}}
\vspace{1mm}
Of course, real physics  is several orders of magnitude more complicated than that and the above
table should be understood as what it is: goofing around with objects and jargon from the standard model. 
But one can often learn from children stories. In any case, toying with ``particle physics" on a graph
can be an educational laboratory for experimentation because we can for example look at how fast different
``particles" move in the time dependent space which is created by the collection of solution paths of the 
wave equation. Since higher dimensional forms travel slower, they 
have more mass. In the discrete, the only input is a graph. Like in ``Schild's ladder" which starts
with {\it In the beginning was a graph, more like diamond than graphite}.
While the present article is mathematical, we believe that the physics could be
interesting, at least as a play-field for experimentation and where the mathematics is not too complicated,
using undergraduate concepts of linear algebra and differential equations only. There is almost no imput. Only 
a graph. No forces, potentials or Lagrangians have to be fed in. All the interactions appear when following 
the integrable deformation $D(t)$. There are many games one can play here. One possibility is to 
avoid looking at the vector space on which $D(t)$ acts but take instead seriously
the columns of the Dirac operator itself. This also simplifies the concept even more. Lets take
for example a column $f$ in the second block of $D(t)$, which is a fermion column. Now, 
$A=Df$ is a $1$-form because $D^2=L$ is block diagonal. Call $F=dA$ the electromagnetic field
and $j=d^* F$ the current. Since time is not part of the graph, (it is implemented by the symmetry $U(t)$),
the evolution $D(t)$ determines a motion of the fields $F(t)$ and currents $j(t)$. Because $D(t)$
asymptotically moves along the wave equation, the fields behave, as physics demands it. Now all the 
columns in the fermion block produce different fields at time $t$, depending on the location on
the graph. The manifold of all solutions of the wave equation which is a closed, compact submanifold of 
$SU(v)$, can be equipped with a Lorentzian structure bringing back the symmetries in the continuum. 
Having reversed relativity and split time away from from space not only makes experimentation easier, 
it also keeps us in the ``playground". Otherwise, we would have to consider non-compact graphs or 
study global variational problems to select out suitable graphs as space-time and also deal with 
discrete time, which is much more difficult, in virtually all situations we have encountered in mathematics. 
The later would not be impossible, because the Euler characteristic of a 
graph is an interesting functional which at least for even-dimensional geometric graphs behaves very much like 
the Hilbert action. It is a sum over all vertices and at every vertex a sum over all possible two-dimensional 
sectional curvatures, where all terms are defined graph theoretical. The reason for this are Gauss-Bonnet-Chern
\cite{cherngaussbonnet}
and index results \cite{poincarehopf,indexexpectation,indexformula} for graphs but the upshot is that Euler 
curvature in the discrete has very much in common 
with scalar curvature at least conceptionally: Ricci and scalar curvature can be written as an average over
all sectional curvatures of planes passing through a line or point, Euler curvature - the integrand in Gauss-Bonnet-Chern- 
can be written as a more exotic average of all possible sectional curvatures through a point. For Riemannian manifolds, this
has not yet been written down but if the analogy should go over, Euler curvature of a point $x$ in an even dimensional 
Riemannian manifold is the expectation value of curvature over all two dimensional embedded subsurfaces (strings) 
passing through $x$ or that Euler curvature $K(x)$ of a point $x$ in a Riemannian manifold $M$ is the average of indices $i_f(x)$ 
averaged over a probability  Morse functions $f$ on $M$. 
The problem with proving this is to have a good probability space of all two dimensional submanifolds
of a compact Riemannian manifold or an intrinsic probability space of all Morse functions on $M$. 
(For the index averaging result, analysis similar to \cite{Banchoff67} show that 
linear functions in an ambient flat Euclidean space work to induce Euler curvature). 
But if the surface curvature average interpretation is true too, 
we can think of Euler characteristic (the average of Euler curvature) 
as a quantized natural functional playing the role of the Hilbert action (the average of scalar curvature). 
In the graph case, the mathematics is much easier and graphs with extremal Euler characteristic should 
play a special role. 

\section{Introduction}

The Dirac operator $D=d+d^*$ for a finite simple graph $G=(V,E)$ or Riemannian manifold $(M,g)$ 
encodes the geometry of a graph or manifold. In the graph case, $(D^2)_0=L_0=B-D$ 
contains all the information about the graph like in the manifold case where 
the operator $L_0$ determines the metric $g$. 
The operator $D$ is defined by the exterior derivative $d$ on the geometry. We look here at isospectral 
integrable systems $$ D' = [B,h'(D)] $$
with $B=d^-d^*$, where $h$ is a polynomial. Any of these systems deform the operator $D$ and so $d$ but do 
not alter $L=D^2$. The operator $D=d+d^*+b$ gains a block-diagonal part $b$ which leaves $p$-forms invariant 
and which leads to a decomposition $D^2=C^2+V$ with $C=d+d^*$ and $V=b^2$. 
All these systems lead to deformations of the geometry because the new Dirac operator $C(t)$ 
can be used to define new distances. It is custom to rewrite such systems in a Hamiltonian form 
$\dot{x}=J\nabla H(x)$.  With $JA:=[B,A]$ and $\nabla h(D):=h'(D)$ 
and Hamiltonian $H(D)=\tr(h(D))$, the system can then be rewritten as
$D'=J \nabla H(D)$. This Lax pair language and the corresponding Hamiltonian formalism which comes with it,
is common in virtually all known integrable Hamiltonian systems. 
The integrable system we consider here, is mathematically close to the Toda lattice \cite{Toda,Teschl} 
$\dot{a}_n = a_n(b_{n+1}-b_n) \dot{b}_n = a_n^2-a_{n-1}^2$ which is a discretisation
of the Korteweg de Vries partial differential equation and 
which is an isospectral deformation $\dot{L}=[B,L]$ of a Jacobi matrix $L$. The motion induces
a Volterra system $\dot{d}_n = d_n(d_n-d_{n-1})$ for a Dirac operator $D$ (also Jacobi but without diagonal part) 
satisfying $D^2=L$ obtained by doubling the lattice. 
Our system is different in that the Laplacian $L$ does not move. Only its square
root $D$ moves and $D$ develops a block diagonal part $b(t)$ which eventually dominates. 
Still, much of the mathematics is related since both use the Lax pair formalism \cite{Lax1968}. 
The Toda system on a linear graph was integrated using scattering methods 
in \cite{Mos75a}. On a circular graph, the Toda system can be conjugated to a translation on a torus
using algebro-geometric tools \cite{Moerbeke79}. There, the Abel-Jacobi map linearizes the motion of 
divisors given by the eigenvalues of an auxiliary spectral problem. 
For the deformation discussed in the present paper, we have in the real case a scattering situation 
and especially do not have recurrence. This makes the analysis easier. 
If the operator $B$ is replaced by $d-d^* + i b$, then $D(t)$ still converges to the same block diagonal operator
$b(\infty)$ and $B(t)$ converges to $i b(\infty)$. 
The unitary flow $U(t)$ satisfying $U' = B U$ consequently an attractor, on which the dynamics is a 
linear wave dynamics $\exp(i t b(\infty))$ and almost periodic. \\

How did we get to the system? Isospectral deformations of higher dimensional Schr\"odinger operators are
in general not possible by a rigidity result of Mumford \cite{MoerbekeMumford}. 
While doing isospectral deformations of 
$L$ is no problem -  any system $\dot{L} = [B(L),L]$ with some antisymmetric $B(L)$ allows to do that -, 
the corresponding unitary evolution does not have the property that $L$ stays a Laplacian. In other words, 
if $L$ is a band matrix, then with such a naive deformation, $L(t)$ is no more a band matrix in general for any $t>0$. 
We have looked at this question in \cite{Kni94} and seen that sufficient for a deformability is a factorization 
$L=D^2$ \cite{Kni94}. Factorization therefore is a condition which naturally leads to the deformation of the 
Dirac operator $D$ because $D$ is the square root of the Laplace-Beltrami operator. 
Having worked with the Toda lattice before \cite{Kni93diss} in an infinite dimensional setting, we 
would have expected at first a recurrent flow for $D(t)$ in the case of graphs like the circular graph and scattering 
situation for the line graph. But this is not the case. The main features of the dynamical system is independent of the graph
and the complex parameter $\beta$. The later only determines whether $B(t) \to 0$ and whether the limiting unitary flow
$U(t)' = B(t) U(t)$ is nontrivial or not. 
On the orthogonal complement of the kernel, the system decomposes into invariant two-dimensional planes. 
On these planes, a scattering motion takes place.
When allowing a complex structure to evolve, that is if $\beta \neq 0$,  we asymptotically get an almost periodic unitary flow. 
The analysis is essentially the same for graphs or for Riemannian manifolds, even so in the later case, 
we have an infinite dimensional situation. On the technical side, we need analytic continuation of the Zeta function 
to define the Hamiltonian in the Riemannian case, but we do not need to know the Hamiltonian at all, 
to define the flow. The differential equations are defined as they are. 
Still, dealing with the Zeta function for the Dirac operator is quite pleasant because unlike the zeta function of the 
Laplacian, it can be chosen so that it has an analytic continuation onto the complex plane, at least for odd dimensional
manifolds. For graphs, the zeta function is of course always analytic everywhere. \\

The deformed operators $D(t)=d(t)+d(t)^*+b(t)$ satisfy $D(t)^2=L$ and 
$d(t) d(t)=0$, so that $d(t)$ remains an exterior derivative. As we will see, the
cohomology groups defined by the exterior derivative $d(t)$ are preserved because cocycles or coboundaries get 
deformed by the differential equation $\dot f = b(t) f$. If $D(t) = U(t)^* D U(t)$, then also the McKean-Singer formula 
$\str(U(t)) = \chi(G)$ holds. While it is not a direct consequence of super symmetry like in 
\cite{knillmckeansinger}, it follows directly from the McKean-Singer result.
While the eigenvalues of $D(t)$ still pair up as in the case $t=0$, the
corresponding eigenfunctions do no more honor the orthogonal decomposition into fermions 
and bosons. The deformation is of a scattering nature since
the operators $D(t)$ converge for $t \to \pm \infty$ to block diagonal operators $b^{\pm}$ 
preserving the linear spaces $\Omega_p$ and which both satisfy $V = b_{\pm}^2=L$. 
While each $D(t)=d(t)+d(t)^* + b(t)$ defines exterior derivatives $d(t)$, the operator $M(t) = (d(t)+d(t)^*)^2$ converges to 
zero in the graph case. Any deformation starting at the original $D$ increases therefore the
Connes pseudo distance between vertices. This happens initially with an inflationary fast start.
While the deformation does not change the original operator, it does change a decomposition:
more and more of the ``kinetic interaction part" $M(t)$ becomes ``potential self-interaction energy" $V(t)=b(t)^2$. 
It changes the relation between position $u(t)$ and velocity $u'(t)$ solving the wave equation 
$u_{tt} = -L u$. A wave with a given initial frequency will later have a smaller frequency 
and will appear red shifted. Since the decay of $C(t)$ is asymptotically exponential, the amount of red shift
will after some time be proportional to the distance traveled. \\

The geometric evolution of $D(t)$ is the same for Riemannian manifolds or graphs. Even the formalism does
not change. We know that the flows are isospectral, and fix the Laplace-Beltrami operator $L$ on $p$-forms.
The new metric defined by $d(t)+d(t)^*$ stays Riemannian, because geodesics still exist and
the polarization identity $f(u+v) + f(u-v) = 2 f(u) + 2 f(v)$ holds with
$f(v) = (d/dt d(exp_x(tv),x))^2$. One can get the metric $g^{ij}(x) = - (M(t)f)(x)/H_{ij}(x)$, 
where $H_{ij}(x)$ is the Hessian matrix of $f$ at $x$. Because $M(t)$ is not ispospectral to $L(0)$, this 
does not contradict spectral rigidity like the Guillemin-Kazhdan theorem 
which tells that compact manifolds of negatively curvature are spectrally isolated 
\cite{Rosenberg,BergerPanorama}. In the continuum, one would have to look at the flow in some
Banach space of pseudo differential operators. Unlike for other known integrable PDE's, the deformed
operator $D(t)=d(t) + d(t)^*$ is no more a differential operator and setting up a natural 
functional analytic framework can be a bit tricky. We do not address this here. 
The flow in the space of metrics on $G$ or $M$ changes the geometry. The long term behavior of the 
metric, when suitably scaled, is not investigated yet, but could be interesting. We can ask for example 
whether it is true that a simply connected compact Riemannian manifold converges to a sphere 
under the evolution. A positive answer would provide a new approach to the Perelman theorem. 
Currently, this is unexplored and we do not know at all, whether a rescaled geometry converges
to a limiting shape. Our analysis does not tell even yet whether we have recurrence or whether
we have a transient situation for the rescaled geometry. This question is independent of $\beta$. 
The limiting shape could be something more general than a manifold or graph. 
In any way, the integrability of the flow prevents trajectories to run into any singularities,
so that unlike the Ricci flow, the geometric evolution should exist for all times in any space of
operators in which the differential equations can be defined and especially preserve categories of 
smoothness which is present initially. 
The flow exists in any function space which is obtained by completing the span 
of finite sums of eigenfunctions. On every plane spanned by $f,g=D(t)f$ especially, 
the flow is given by a time-dependent ordinary differential equation. 
$\dot{f} = a(t) f + b(t) g, \dot{g} = c(t) f + d(t) g$
where $a(t),b(t),c(t),d(t)$ are globally bounded. 
The deformation provides an infinite-dimensional family of metrics on $M$.
In the continuum, the isospectral deformations are KdV type partial differential equations despite the
fact that we deal with pseudo differential operators. 
Since there are invariant McKean-Singer planes, it not only immediately establishes that the ordinary 
or partial differential equation under considerations have solutions for all times; 
it also immediately suggests how to make finite-dimensional Galerkin approximations: 
we can look at the invariant space of a finite set of eigenvectors 
closed under the McKean-Singer map $v \to D(t) f$ which has the property that it leaves
the McKean-Singer planes invariant. We have tried this out for the circle but how well the 
chosen Galerkin method mirrors the infinite dimensional dynamics is not investigated yet. 
In the continuum, the super trace of the unitary evolution $U(t)$ must be either
defined by analytic continuation or as a limiting case of finite dimensional approximations $U_n$,
where $U_n$ is the evolution defined on finite dimensional invariant subspaces built up by 
McKean-Singer planes. The continuum could be linked to the graph case also by a 
limiting procedure like for Hodge Laplacians \cite{Mantuano}.
The Mantuano paper suggests that if we make a fine enough triangularization of a manifold
and look at the flow on the graph, then the graph evolution should be close to the flow on the manifold. 
Especially, we could study the evolution of the manifold $M$ by evolving graphs belonging 
to finite triangularizations of $M$. 

\section{Analysis of the flow} 

The analysis for the differential geometric and graph theoretic case are similar. We focus on 
the graph case, where everything is finite dimensional. We also look mainly at the real evolution. 
We comment on the complex evolution in a different section and plan to extend the analysis a bit more 
elsewhere. It is exciting because the emergence of a complex structure is interesting by itself, 
leading to discrete Dolbeaux type cohomologies. For the flow, the different $\beta$ will essentially 
just produce a time change. The paths of $V(t)$ and $M(t)$ which build up the Laplacian $L=V(t) + B(t)$ 
are independent of $\beta$. \\

Let us start with a review on the definition of $D$. We look at the set $\G$ of all complete subgraphs of $G$. It is a graph 
by itself, where two simplices $x,y$ are connected if $x$ is contained in $y$ or $y$ is contained in $x$ and if 
the dimensions of $x$ and $y$ differ by $1$. We now equip the graph $\G$ with an orientation, a choice of a permutation
of the vertices $(x_0,...,x_n)$ of each $x \sim K_{n+1}$. The symmetric matrix $D$ is defined by $D_{ij} = 1$ if $i \subset j$
and the permutation of $j$ restricted to $i$ has the same sign as the permutation given on $i$. The same is done if the roles of $i,j$ 
are reversed. Otherwise, if the signs do no match, we have $D_{ij}=-1$. 
The choice of signs or ``spin" corresponds to a choice of basis or gauge and is irrelevant for all considerations. Different 
orientation choices lead to unitary equivalent matrices $D$. If we write $f(x_0, \dots, x_n)$ for a function 
$f$ on the simplex $K_{n+1} = (x_0, \dots , x_n)$, ordered according to the choice of the orientation, we have
$f(\pi(x_0, \dots , x_n) = {\rm sign}(\pi) f(x_0, \dots, x_n)$ for any permutation of the $n+1$ vertices. We can look at $f$
as a function on the simplices.  The operator $d$ is then the exterior derivative 
$df(x_0,\dots,x_n) = \sum_{k=0}^{n} f(x_0, \dots, \hat{x}_k, \dots, x_n)$ and the Dirac operator is the symmetric matrix
$$ D = d+d^* \;  $$
of size $v \times v$ matrix, where $v = \sum_{k=0}^{\infty} v_k$ and $v_k$ is the number of $k+1$ simplices
in $G$. Its square $L=D^2=d d^* + d^* d$ is the discrete Laplace-Beltrami operator of the graph and sometimes also called
Hodge Laplacian. Unlike $D$, it leaves the space $\Omega_k$ of $k$-forms invariant. 
By scaling the exterior derivatives $d_k$ as $\gamma_k d_k$ with real nonzero $\gamma_k$, we could generalize the 
Dirac operator more. These changed operators are not unitarily equivalent but essential features like symmetries
remain the same. The constants correspond to units used on the different $\Omega_k$ subspaces. The constants 
influence the evolution. Since the evolution is linear, already scaling the entire operator by a constant has
drastic effects. \\

The definition of the dynamical system is the same if $(M,g)$ is a compact Riemannian manifold
and where the exterior derivative $d$ defines a self-adjoint operator $D=d+d^*$.
In the case of Riemannian manifolds, a standard initial functional analytic setup called elliptic regularity is needed
which assures that $D$ and $D^2$ have discrete eigenvalues. Since we look at isospectral deformations, we could for
the linear algebra part restrict to the eigenspaces belonging to eigenvalues smaller than some constant $\lambda$ and 
deal with finite-dimensional matrices also in the manifold case. While higher energy eigenfunctions still influence
the dynamics on the low energy McKean-Singer planes, their influence will become smaller and smaller for $\lambda \to \infty$.
If we look at the dynamics on a finite time interval $[0,T]$, then the orbits of the Gelerkin approximation converges,
as long as we work with operators on a function space in which every $f$ has an expansion $\sum_n a_n f_n$ 
with eigenfunctions $f_n$ of $L$ such that $a_n \to 0$.  \\

Given $D=d+d^*+b$, define $B=d-d^* + \beta i b$. At $t=0$, we have $D=d+d^*$ and $B=d-d^*$. 
If $\beta=0$, then the matrices $D(t),B(t)$ stay real. Any of the
Lax pairs $D'=[B,h(D)]$ lead to a system of differential equations. It is custom to write such differential equations
in Hamiltonian form $D'=J \nabla H$, where $H(D)=\tr(h(D))$ and $\nabla H(D) = h'(D)$ and $J(X) = [B,X]=BX-XB$.
Alternatively, one can formulate the dynamics using Poisson brackets 
$\{ F,G \; \} = \tr(\nabla F(D) B \nabla G(D)) = 2 \tr(\nabla F(D) J \nabla G(D))$ as $F' = \{ F,H \; \}$
for observables $F,G$ which real valued functions, where $F(D),G(D)$ is defined by the functional calculus. 
For example, for $F(x)=x^n$, since all traces $\tr(D^n)$ are invariant, 
the Poisson brackets $\{ F_n,F_m \; \}$ are zero for $n \neq m$. 
Because $D^2=L$ and $[B,L]=0$, we see that any flow can be written as $D'=f(L) [B,D]$. Since
the flow leaves planes $E_{\lambda}$ spanned by eigenvectors $f,Df$ of an eigenvalue $\lambda$ 
invariant, higher flows just involve an energy dependent time change on each of these planes. 
The operator $D_{\lambda}$ obtained by restricting $D$ to $E_{\lambda}$ satisfies 
$D'_{\lambda} = f(\lambda^2) [B,D]$. We therefore stick to the first flow with Hamiltonian $\tr(L^2)$. \\

As custom for integrable Lax systems, one considers also the unitary evolution defined by 
$U' = BU$. It satisfies $D(t) = U D(0) U^*$ because $U^* B D U + U^*[B,D] U + U^* D B U = 0$. 
The spectrum of $D(t)$ therefore stays the same.
All flows commute. The deformed operator $D(t)$ has the form $D(t) = d(t) + d(t)^* + b(t)$, 
where $b(t)$ is block diagonals, in the blocks, where the Laplacian $L$ has nonzero parts. 
The flow defines a deformation of the exterior derivative. This is by itself nothing
special because one could achieve deformations of the exterior derivative in many other ways. 
A particular useful one is the deformed Laplacian \cite{Witten1982,Cycon}
using $e^{-f t} d e^{f t}$ which provided a new approach to Morse theory. 
The deformed exterior derivatives are not isospectral.
A detailed study the spectrum of the deformed Laplacian is 
semi-classical analysis, which is quite technical. We expect therefore also that 
the task to be nontrivial to find the motion of the spectrum of $C(t)=d(t) + d(t)^*$ in detail. \\

The deformation $\dot{D}=[B,D]$ produces a family $D(t) = d(t) + d(t)^* + b(t)$ of operators, where
$B(t) = d(t) - d(t)^*+i \beta b(t)$. We call it the Dirac deformation.
Define also $C(t) = d(t) + d(t)^*$ and $M(t) = C(t)^2$ and $V(t) = b(t)^2$.  We have $B^2=-M$ because
$(d-d^*)^2=-(d+d^*)^2$. We also have $B^2 = -L$ because $B b + b B=0$. The operator $A=B/i = (d-d^*)/i + \beta b$ 
and $D$ are both square roots of $L$ but they do not commute. \\

Given an eigenvector $f$ of $L=D^2$ to a nonzero eigenvalue $\lambda$, we call the plane spanned by $f$ and $Df$ a
McKean-Singer plane. The matrix $M$ is the square of a symmetric matrix and has nonnegative 
eigenvalues too. Written out for $\beta=0$, the Lax pair for the first flow is $d'=[d,b],b'=[d,d^*]$ or 
\begin{eqnarray*}
d'   &=& d b-b d        \\
(d^*)' &=& b d^* - d^* b  \\
b'   &=& d d^* - d^* d  \; . 
\end{eqnarray*}
To see this, just write out $D' = B D - D B = (d-d^*) (d+d^*+b) - (d+d^*+b)(d-d^*)$
and use $d d= d^* d^* = 0$ as in the computation done in \ref{cohomology}). In general, the first line to the right
is multiplied with $(1-i \beta)$ and the second line by $(1+i \beta)$. \\

{\bf Remarks.} \\
{\bf 1)} As mentioned already, more general flows with different Hamiltonian $H$ lead to a multiplication on the right 
hand side with a function of $L$. This affects the dynamics by a time change on each 
McKean-Singer plane. The fact that $L$ commutes with $b$ will make the general case a 
quite obvious modification of the first flow, which therefore displays all essential 
features considered for the first flow.  \\
{\bf 2)} We will comment on the case $\beta \neq 0$ more below. It just produces a complex
time change for $d$ and $d^^*$.  \\
{\bf 3)} The operator $D$ has entries $-1,0,1$ only. If we replace $D$ by $\gamma D$ for 
some constant $\gamma$, then the evolution changes. In general, a larger $\gamma$ will lead
to an initial inflation rate increase which is exponential in $\gamma$.  \\
{\bf 4)} When looking at the equation for $d$, we see a logistic nature, which is common 
in population models. Initially, when $d=0$, there is no linear growth of $d$ because $b$
is zero. At larger times, when $d$ has become smaller and $b$ become larger (they are balanced
by $(d+d^*)^2+b^2=L$ being constant), then again the growth goes to zero. A naive estimate
suggests that the maximal growth is around the time when $d$ and $b$ are balanced. Since $b$ 
settles, the differential equations will show exponential decay asymptotically.  \\

Lets call an eigenfunction $f$ bosonic if $f$ is in $\Omega_b$. 
Eigenfunctions in $\Omega_f$ are called fermionic. The following is formulated for the
finite dimensional graph case: 

\begin{thm}
The Dirac deformation is completely integrable in the following sense: 
The operator $D(t)$ converges for $|t| \to \infty$ 
in the matrix norm to matrices $D(\infty) = - D(-\infty)$. 
Each $D(t)$ leaves the McKean-Singer planes invariant.
If $f$ is a bosonic or fermionic eigenfunction, for which $g(t)=D(t) f$ is 
originally perpendicular to $f$, then $\sin(\alpha(t)) \to 0$, where $\alpha(t)$ is the 
angle between the fermionic space and $g(t)$. The operator $C(t)$ converges to the zero operator 
$0$ in the strong operator topology, with a fast inflationary start.
\end{thm}

\begin{lemma}[Isospectral]
For every $t$, the operator $D(t)$ is isospectral to $D(0)$. 
An eigenfunction of $D$ moves with the differential equation $f'=Bf$.
\end{lemma}
\begin{proof}
The differential equation $U'=BU$ for the unitary $U9t)$ satisfying the initial condition
$U(0)=0$ provides the conjugation $D(t)= U(t) D(0) U(t)^*$. The spectrum of $D(t)$ and $D(0)$ are 
therefore the same. If $f(0)$ is an eigenfunction of $D(0)$ then $U(t) f(0)$ is an eigenfunction 
for $D(t)$. 
\end{proof}

\begin{propo}[Deformed cohomology]
\label{cohomology}
The relation $d(t) \circ d(t) = 0$ holds for all $t$ and the cohomology groups deform in an explicit way:
if $f$ is a cocycle and $f'=-bf$ then $f(t)$ stays a cocycle for $d(t)$. If $f$ is a coboundary and 
$f'=-bf$ then $f(t)$ stays a coboundary for $d(t)$.  If $f'=B f$ and $f$ is harmonic at $t=0$, then it
stays harmonic. 
\end{propo}
\begin{proof}
a) To show $d^2=0$, just differentiate $d^2$ using the Leibniz rule: 
$(d d)' = d' d + d d' = (db-bd) d + d (db-bd) = dbd-dbd = 0$. This shows that we have 
again a cohomology. To show that the cohomology stays the same, we deform the cocycles and coboundaries. 
The computation is the same for $\beta \neq 0$ and shows that the cocycles and coboundaries for $d$ stay
real even so when $d$ becomes complex. \\
b) Assume $df=0$  and $f'=-b f$. We show that $d f_t=0$. 
We have $d/dt (df) = (d b-b d) f -  d b f = b (df) = 0$. Again, for $\beta \neq 0$, we just have additional 
factors $(1 \pm i \beta)$. \\
c) If $f=dg$ and $g'=-bg$ then $f'=-bf = (db-bd) g - d b g = (d g)'$ so that $g(t)$ stays a coboundary. \\
Proof. $(d f)' = d' f + d f'  = db f - b df + df' = d bf + d f' = 0$. \\
d) If $f$ is harmonic $Lf=0$, then the operator $f(t)$ satisfying $f'=Bf$ says harmonic.
The operator $L$ commutes with $U$ because $L$ commutes with $B$.
\end{proof}

\begin{lemma}[The derivative of $B$]
For $\beta=0$, we have $B' =[D,b]$.
In general, $B' = [D,b] + \beta [d,d^*]$. 
\end{lemma}
\begin{proof}
$d'-(d^*)' = (db-bd) - b d^* + d^* b = (d+d^*)b -b(d+d^*)  = Db-bD = [D,b]$. 
\end{proof}

\begin{lemma}[Anticommutations]
For all $\beta$ and all $t$, we have 
$\{d,b \; \} = \{D,b \; \} = \{B,b \; \} = 0$. 
\end{lemma}
\begin{proof}
Initially, we have $\{d,b \; \}=0$. We want to see that this remains the case. 
Use the Leibniz product rule and add
$(db)'=d'b+db' = (db-bd) b - d d^* d$ to 
$(bd)'=b'd+bd' = d d^* d + b (db-bd)$. Using $\{d,b \; \}=0$
we see that the sum is zero. The computation for $\{d,b\}=0$ is similar
so that the statements $\{D,b \; \} = \{B,b \; \}$ follow by linearity. 
\end{proof}

\begin{coro}
We have $\{B,D \; \}= 2 i \beta b^2$ for all times. 
Consequently, the flow can for $\beta=0$ also be written as $D' = 2 B D$. 
\end{coro}
\begin{proof}
Initially, we have $\{B,D \}=\{ d-d^*,d+d^* \} = 0$. Later,
when $B=d-d^*$, $D=d+d^* + b$ we use the previous lemma stating 
that $\{d,b\}$ and $\{d^*,b\}$ are zero for all $\beta$ and all $t$. 
\end{proof}

\begin{lemma}
The symmetric operators $b,R=d d^*,S=d^* d,V=b^2$ all commute and 
can be simultaneously diagonalized.
\end{lemma}
\begin{proof}
All these operators only depend via a time change on $\beta$. 
$S=d d^*$ and $S=d^* d$ commute because their products $RS,SR$ are zero.
Initially, $[b,d d^*]=0$. Now differentiate. We have $[b',d d^*]=0$. We also
have $[b,(d d^*)' = [b,(db-bd) d^* + d (b d^* -d^*b)] = [b,-bd d^*-d d^* b]=0$.
\end{proof}

\begin{coro}
The Laplace-Beltrami operator $L(t)$ does not move.
\end{coro}
\begin{proof}
$(D^2)' = D' D + D D' = [B,D] D + D [B,D] = B L - L B = [B,L]$. 
Since $B^2 = -M$ and $b,M$ and $b,L$ commute, we have $[B,L]=0$. 
\end{proof}

\begin{propo}
We have $L=M+V=R+S+V$, where $M,L,V,R,S$ all pairwise commute, are symmetric and 
$M,V$ both have nonnegative eigenvalues. 
The operators $b,M,L,V$ all have the same kernel for $t>0$.
\end{propo}
\begin{proof}
We have $L=d d^* + d^* d + b^2 + db + d^* b + b d + b d^*$, where the last part 
is zero.  We check that $\{d,b\}=\{d^*,b\}=0$ so that $[d d^*,b] = [d^* b,b] = 0$.
The operators $M,V,L$ all have the same kernel for $t>0$ because they 
commute. 
\end{proof}

Here comes the key lemma: 

\begin{lemma}
$O=b d d^*$ is symmetric with no negative eigenvalues and
$Q=b d^* d$ is symmetric with no positive eigenvalues.
Both operators $O,Q$ are real for all $t$ and all $\beta$. 
\end{lemma}
\begin{proof}
The computations are similar for $O$ and $Q$ and we only write it out for $O$. 
The matrix $O$ is symmetric because $O^* = d d^* b = b d d^* = O$. 
The eigenvalues of $O,Q$ are zero at $t=0$ because $b$ is zero at $t=0$. We look now at how eigenvalues 
of $O$ change and show that at $0$, the eigenvalues of $O$ can not move to the left and the eigenvalues of $P$
do not move to the right. \\
The Rayleigh perturbation formula tells $\lambda' = \langle w O',v \rangle$,
where $v$ is a unit eigenvector of $\lambda$ and $w$ the dual vector $(O^*)' w$ to $v$. By symmetry of $O$ we have 
$\lambda' = \langle  v,O'v \rangle$. We know that $b$ and $d d^*$ have the same eigenvectors for $t>0$
because they commute. \\
First, we get $(d d^*)' = (d b-b d) d^* + d (b d^* - d^* b) = -2 b d d^*$ and from that,
$O' = (d d^*)^2 + b (d d^*)' = (d d^*)^2 - 2 b^2 d d^* = (d d^*)^2 - 2 b O$. 
Assume $v$ is an eigenvector to $O$ to the eigenvalue $\lambda$, then 
$\lambda' = \langle v,O'v \rangle = \langle v,((d d^*)^2 v -2bO) v\rangle$. This means that 
$\lambda' \geq 0$ if $\lambda=0$ meaning that $\lambda$ can not cross to the negative side. 
\end{proof}

We have $\tr(b)=0$ for all $t$ because $\tr(d d^* - d^* d)=0$. 
The next trace of $b$, which is $\tr(b^2)$ turns out to be a Lyapunov function: 

\begin{lemma}[A Lyapunov function] 
$\tr(b^2)$ increases monotonically so that $\tr(b^2)' \geq 0$. 
Initially, for $t=0$, we have $\tr(b^2)'=0$ and asymptotically, we have $\tr(b^2)=\tr(L)$
so that $\tr(b^2)'$ will have a maximum somewhere. 
\end{lemma}
\begin{proof}
Use that $d/dt \tr(b^2) = 2 \tr(b b')=2 \tr(b d d^* - b d^* d)$
and that $b d d^*$ and $-b d^* d$ have both only nonnegative eigenvalues. 
\end{proof}

\begin{coro}
$\tr(M') \leq 0$.
\end{coro}
\begin{proof} 
We have $\tr(bC) = \tr(Cb)=0$ because these matrices do not have anything in the diagonal.
Because $L$ does not move, we have $\tr(L)'=0$ and 
$$  \tr(M(t))' =\tr(L-C b - b C - b^2)'= -\tr(b^2)' \leq 0 \; . $$
\end{proof}

\begin{coro}[Attractor] 
$M(t)$ has its spectrum in $[0,a(t)]$ with $a(t) \to 0$ for $|t| \to \infty$.
$d(t)$ converges to zero and $b(t)$ converges to an operator $V$ satisfying $V^2=L$.
In the graph case, the convergence is in norm, in the manifold case in the 
strong operator topology. 
\end{coro}
\begin{proof}
$M(t)=C(t)^2+B(t)^2$ has positive or zero eigenvalues. We have seen that the trace decreases. 
If we would converge to something different from $0$ then we would have an other equilibrium point.
\end{proof}

We call $g(t)=D(t) f$ the super partner of $f$ if $f$ is an 
eigenfunction of $L$ to a nonzero eigenvalue $\lambda$. By definition, the super partner of
a super partner is a multiple of the original $f$. Because $-B^2 = C^2 = M$, we could 
also have looked at $f'=B f$, which is also perpendicular to $f$ initially. 

\begin{coro}[Superpartners]
If $f$ is an eigenvector of $L$, then both $D(t) f$ and $B(t) f$ are eigenvectors of $L$.
\end{coro}
\begin{proof}
This follows from the fact that both $B$ and $D$ commute with $L$. 
$L f = \lambda f$, then $L D(t) f = D(t) L f = D(t) \lambda f = \lambda D(t) f$. 
In the same way, $L B f = B Lf = \lambda Bf$. 
\end{proof}

\begin{coro}[McKean-Singer planes]
The vectors $f(t),g(t) =D(t) f(t)$ stay in the plane spanned by 
$f(0),g(0)$ for all times $t \in \R$. 
\end{coro}
\begin{proof}
$f,g=D(t)f$ span a two dimensional McKean Singer plane. 
Since $B$ and $L$ commute, $D g$ is again a multiple of $f$ and $g$. 
That means the vector field takes values in the plane spanned by $f$ and $Df$. 
\end{proof} 

\begin{lemma}
If $f(0) \in \Omega_f$, then the angle between $D(t)f(0)$ and $\Omega_f$ 
converges to $0$ for $t \to \pm \infty$. 
\end{lemma}

\begin{proof} 
Let $\Sigma$ denote the two-dimensional eigen space spanned by $f(0),D(t) f(0)$. 
Because $D(t) f \to b_t f$ and $b_t$ preserves the $\Omega_b \oplus \Omega_f$ splitting
and the intersection of $\Sigma$ with $\Omega_f$ is one dimensional, the angle between $\Omega_f$
and $D(t) f(0)$ has to converge to zero. 
\end{proof}

Define $R=d d^*$ and $S=d^* d$.  

\begin{lemma}
The operator $R b$ and $S b$ both leave the McKean Singer planes invariant. 
Both $R b$ and $S b$ restricted to a two-dimensional McKean Singer plane have one 
nonzero and one zero eigenvalue.
\end{lemma}

\begin{proof}
We know that $Rb,Sb$ and $B$ do leave the plane invariant so that $DB b=(R-S) b$ does. 
We need explicit kernel elements: if $f,g=Df$ span the McKean-Singer plane,
we have to show that $Rb f$ and $Rb Df$ are parallel or that the matrix
$$  \left[ \begin{array}{cc} \langle Rb f,f \rangle     & \langle Rb f,g \rangle \\
                             \langle Rb g, f \rangle  & \langle Rb g,g  \rangle 
           \end{array} \right] \;  $$
has a zero determinant $\langle Bf,Rf \rangle   \langle Bg,Rg \rangle   
                   =   \langle  bf.Rg \rangle   \langle bg,Rf \rangle$. 
Assume we had $Rb f=0$ and $Rb g=0$. Then  $d d^* f=0$ and $d d^* (d+b) f=0$
which implies $d d^* f=0$ and $d d^* d f=0$. But we know that $g = d^* d f$ can not be zero
because otherwise, $f$ would be a new harmonic which is not possible. But now $d^* g=0$ and $d g=0$
but this means $g$ is harmonic and so $d^* d d^* d g=0$ but then it is in the kernel of $d^*d$ also
because the matrix $d^* d$ is symmetric. This contradiction shows that $Rb$ can not be the zero matrix.
The computation for $Sb$ is similar. 
\end{proof}

None of the isospectral flows have an equilibrium point for which $d(t)$ is not equal to zero: 

\begin{coro}[Lack of equilibria]
$M(t)$ goes to zero for $|t| \to \infty$ but is not zero for finite $t$. 
\end{coro} 

\begin{proof}
Lets look at the first flow. The higher flows are just a time change on each McKean-Singer plane. 
If $\tr(M')=0$ then $b d d^*=0$ which is only possible at $t=0$ or for $|t| \to \infty$. 
If $b$ were zero at some time $t_1$, then $M=L$ which is not possible because
$M(t)$ and $d(t)$ converge to $0$ and $b(t)$ converges to $V$ satisfying $V^2=L$. 
We have seen that $M(t)=C(t)^2$ has positive or zero eigenvalues and that the trace decreases.
If $\tr(b^2)'=0$ then $b d d^*$ and $-b d^* d$ must both be zero which is not possible. 
Here is an other argument: 
Since $[B,D]=[d-d^*,d+d^*]=2 d d^* - 2 d^* d$ is never zero as we can see
when we apply $[B,D]$ to $n$ forms, where $[B,D] = 2 d d^*$ is never zero because
this is $2L_n$ and this being zero would mean that the matrix is zero which is not
possible because the diagonal is not for $t=0$.
Because $[B,L]=0$ and $[L,D]=0$ we have $[B L^k,D] = 2 L^k( d d^*-d^* d)$ and also
this can not be zero as an operator.
\end{proof}

\begin{lemma}[Asymptotics]
We have $D(t) + D(-t) = 2C(t)$. In particular $D(\infty) = - D(-\infty)$. 
\end{lemma}
\begin{proof}
Initially, at $t=0$ this is true. Now check that $b(t) + b(-t)=0$ for all $t$. 
\end{proof}

\section{Non-linear McKean-Singer symmetry}

Now, we look at the nonlinear analogue of the McKean-Singer symmetry.

\begin{lemma}
For every $k>0$ we have $\str(B^k)=0$ for $t=0$ and 
${\rm Re}(\str(B^k))=0$ for every $t$. 
\end{lemma}
\begin{proof}
For odd $k$, there is nothing real in the diagonal and the claim is trivial.
For even $k$, note that $B^2=L$ at all times, so that $B^{2n}=L^n$. But we know
that $\str(L^n)=0$ by the classical McKean-Singer result.
\end{proof}

{\bf Remark.} \\
For $\beta \neq 0$, $\str(B)$ becomes purely imaginary for nonzero $t$. \\

The following nonlinear analogue of the McKean Singer result
needs analytic continuation in the Riemannian geometry case because $U(t)$ is not 
trace class in the continuum. Lets focus on the graph case where we deal with finite
matrices. We prove that $\str(U(t))$ remains constant for 
for all complex $t$ if $\beta=0$. 

\begin{propo}[Mckean Singer corollary]
a) For all $\beta$, we have ${\rm Re}(\str(U(t))) = \chi(G)$ for all $t$. \\
b) The trace of $U(t)$ stays real for all $t$. 
\end{propo}

\begin{proof}
a) We know that $\str(L^n)=0$ initially because of the classical McKean Singer symmetry. 
Because $\str(U(t))$ is analytic, we have to show that the real part of 
$d^k/dt^k \str(U) = 0$ for all $k$ at $t=0$.
Differentiate $U$ at $t=0$. $U'=BU, U''=B' U + B^2 U=B U^2+B^2 U, U'''=B U^3+B^2 U^2 + 2 B^2 U + B^3$ etc.
Using $U(0)=Id$ and the lemma implies that all these derivatives are zero at $t=0$.  \\
b) The second statement follows from the fact that if $\lambda$ is an eigenvalue of $U(t)$ then 
also $\overline{\lambda}$ is an eigenvalue of $U(t)$ because $U f = \lambda f$ implies 
$\overline{U} \overline{f} = \overline{\lambda}  \overline{f}$ and the fact that the evolution 
with $+\beta$ gives an isospectral evolution to the evolution with $-\beta$. 
\end{proof}

{\bf Remarks.} \\
{\bf 1)} For $\beta \neq 0$, the super trace $\str(U(t))$ of $U(t)$ becomes imaginary and oscillates. 
We still have ${\rm Re}\str(U(t)) = \chi(G)$ for all $t$. 
Similarly, the trace $\tr(U(t))$ of $U(t)$ becomes imaginary and oscillates. \\
{\bf 2)} There are various discrete symmetries in the system. 
Besides the "charge" symmetry $C: \lambda \leftrightarrow -\lambda$,
the "super symmetry" preserving $\str(U(t))$, there is a "time reversal symmetry"
$T: U \leftrightarrow U^*$. There is also the symmetry $\beta \to -\beta$
and $D=d+d^* \to A=(d-d^*)/i$. Applying the symmetry again, gives $-D$. \\
{\bf 3)} The last remark suggests to write $A=j D$ where $j$ is the $j$ in a quaternion $z=a+bi+cj+dk$. Then 
$i A= i j D = k D$ so that $D'=[k D,D]$ and $B=k D$, $A=j D$. 
One can write the Lax pair as an equation for one operator $D$ using quaternions. The solution $D(t)$
must be understood as a quaternion then but this does not save us any computation time since it does
not save us to store the Lax pair $D$ and $B$, which now appear as rotated by 90 degrees in a quaternion
algebra. It just shows that the nonlinear equation is natural. 

\section{Complex case} 

The analysis so far was done mostly in the case $\beta=0$. Essential features stay the same
when turning on the complex parameter $\beta$. It turns out that the dynamics $D(t)$ is only affected by 
a time re-parametrization. What changes however is that $U(t)$ does not converges to unitary 
operators $U(\pm \infty)$ for $t \to \pm \infty$ as in the real case, but that $U(t)$ approaches 
an almost periodic attractor which describes the linear wave evolution asymptotically. \\

{\bf Remarks. } \\
{\bf 1)} The change from $B=d-d^*$ to $B=d-d^*+i \beta b$, where $\beta$ is a parameter, is similar 
to \cite{Toda} who modified the Toda flow by adding $i \beta$ to $B$.  \\
{\bf 2)} For $\beta=1$, we have $B^2=-C^2-b^2=-L$ and $A=i (d-d^*)$ satisfies $A^2=M$, and from $\{d-d^*,D \;\}=0$
which follows from $\{d,b\}=\{d^*,b\}=0$ we see that
for $\beta=0$, the super symmetry relations $\{A,D\}=0, A^2=M, D^2=L$ hold between
the self-adjoint operators $D,A=B/i$ and $L$. So, $\beta \neq 0$ changes some symmetry.  \\
{\bf 3)} In general, independent
of $\beta$, the time evolution has changed some symmetry already:
the relation $\{P,D\}=0, P^2=1, D^2=L$ for $t=0$ has led to the $\lambda \leftrightarrow -\lambda$ 
symmetry of the spectrum of $D$ because if $f$ is an eigenvector, then $Pf$ is an eigenvector. 
This analysis with $P$ only works for $t=0$ because $\{P,D \}=0$ is false in 
general for $t \neq 0$. But since the flow is isospectral, the charge symmetry 
$\lambda \leftrightarrow -\lambda$ still holds for all $t$.  \\
{\bf 4)} In the real case $\beta=0$, we have for all $t$:
If $f$ is an eigenvector of $D$ to the eigenvalue $\lambda$, then $Af$ is an eigenvector to 
the eigenvalue $-\lambda$. In the same way, if $f$ is an eigenvector of $A$ to the eigenvalue $\lambda$
then $Df$ is an eigenvector of $A$ to the eigenvalue $-\lambda$. \\
{\bf 5)} Since $D(t)$ is now complex, even so it is selfadjoint, the eigenvectors have become
complex. (Similar than for $\left[ \begin{array}{cc} 0 & i \\ -i & 0 \end{array} \right]$
which has the eigenfunctions $\left[ \begin{array}{c} \pm i \\ 1 \end{array} \right]$ to the eigenvalues $\pm 1$). \\
{\bf 6)} In the complex case, where $B$ also has a diagonal part, not only the {\bf super partner} $D f$ but also 
the choice of the {\bf anti-particle partner} $Bf$ is not perpendicular to $f$. Because the eigen-space of $-\lambda$
is also at least 2 dimensional, there are many anti particle partners and because the evolution is unitary, there are
anti particle partners of $f$ which are perpendicular of $f$. \\
{\bf 7)} At time $t=0$, we have $B D=-D B = (d-d^*) (d+d^*) = d d^* - d^* d$. \\
{\bf 8)} The symmetry $A \leftrightarrow C$ is an other symmetry to consider. Because both $A$ and $C$ are square roots of $M$
and both feature supersymmetry, the spectra of $A$ and $C$ are the same. But because they have not the same eigenfunctions
(they do not commute), they can not be diagonalized simultaneously.   \\

The essential features of the system like expansion are $\beta$-independent. Actually, 
$\beta$ just produces an additional "force" onto the motion of the diagonal energy part $b(t)$ of $D$
which accelerates the convergence towards the attractor:

\begin{lemma}[Time change]
The following statements hold for all $\beta$. \\
a) The operators $b(t), d d^*, d^* d, M(t) = d d^* + d^* d$ are real and move along paths which are independent of $\beta$. \\
b) The acceleration ratio is $b''_{\beta}(t)/b''_0(t) = 1+\beta^2$. 
\end{lemma}
\begin{proof}
Since $b'$ is real, $b$ must stay real assuring that $B(t)$ stays skew symmetric and $U(t)$ stays
unitary. Now, we want to see what effect it has if we change $\beta$.
The equations of motion $(d + d^* + b)' = [d-d^*+i \beta b,d+d^*+b]$ can be rewritten as
\begin{eqnarray*}
d' &=& 2 (1-i \beta)  db \\
(d^*)' &=& 2 (1+i \beta) b d^* \\
b' &=& d d^* - d^* d  \; . 
\end{eqnarray*}
In order to see that $b$ is up to a time change, it is enough to see that $d d^*$ 
and $d^* d$ just have a $\beta$-dependent time change. Indeed, $d d^*$ and $d^* d$ and $b$ are all real and
$b'' = (d d^*)' -(d^* d)' = 4 (1+\beta^2) d b d^*$. \\
Since $b^2 + M = L$ does not move, also $M(t) = L-b^2(t)$ traces a path which is independent of $\beta$. 
\end{proof}

\begin{coro}
The limiting operator $D(\infty)$ is independent of $\beta$. 
\end{coro}
\begin{proof}
Having the same curves, this means that it is only the rate of convergence which 
depends on $\beta$. 
\end{proof}

Remarkably, for $\beta \neq 0$, a complex differential structure has emerged during the evolution: 
the operator $D(t)$ has become complex for $t>0$, even so we have started
with a real graphs or manifold. Define $\partial={\rm Re}(d)$ and
$\overline{\partial}=i{\rm Im}(d))/2$, then $\partial^2 = \overline{\partial}^2 = 0$ and
$d=\partial + \overline{\partial}$. Because cocycles and coboundaries deform in an explicit way,
cohomology groups defined by $\partial$ and $\overline{\partial}$ are both the same than for $d$. \\

Since $\partial \overline{\partial} = \overline{\partial} \partial=0$, the Laplacian $L=D^2$ is
the sum of two Laplacians $L^{\partial} = (D^{\partial})^2$ and
$L^{\overline{\partial}} = (D^{\overline{\partial}})^2$, where $D^{\partial} = \partial + \partial^*$
and $D^{\overline{\partial}} = \overline{\partial} + \overline{\partial}^*$. 

The complex structure disappears asymptotically for $t \to \infty$. 
For large $t$, the linear flow $\exp(i D(\infty) t)$ is close to the nonlinear flow 
$U(t)$, satisfying $U(t)^* D(t) U(t)=D(0)$. \\
While the exterior derivative $d$ as well as the new Dirac operator $C=d+d^*$ are complex for $t>0$, the
operator $M=C^2$ is real if we start with a real $D$. While asymptotically,
$||{\rm Im}(D(t))||/||{\rm Re}(D(t))|| \to 0$, the complex structure is especially relevant in the
early stage of the evolution.

\section{Example: the circle}

The simplest compact Riemannian manifold without boundary is the circle $M=T$ with the standard
homogeneous metric $g_{11}=1$. The Dirac operator is a differential operator on $\Lambda M$ which is a 
$2^{dim(M)}$-dimensional vector bundle over $M$. With $df(x) = f_x dx$ and $d^* g(x) dx = -g_x$, the Dirac operator is
$$ D = \left[ \begin{array}{cc} 0           & \partial_x \\
                               -\partial_x  & 0 \end{array} \right] \; . $$
Its square is the Laplace-Beltrami operator
$$ L = \left[ \begin{array}{cc} -\Delta  &  0 \\
                                  0      &  -\Delta \end{array} \right] \;  $$
which respects the split into zero and one forms. Fourier theory diagonalizes these operators. The eigenbasis 
$\left[ \begin{array}{c} \pm i e^{i n x} \\ e^{i n x} \end{array} \right]$
belongs to eigenvalues $\pm n$ so that the spectrum of $D$ is the set of integers $\Z$ and $L$ 
has the eigenvalues $n^2$ for $n=0,1,\dots$.  \\

The letters $A,B,C$ in the following have no relation with previous use of $A,B,C$ in this text. 
We look here only at $\beta=0$: 
the deformation $D=D(t)= d+d^* + b = \left[ \begin{array}{cc} B           & A                    \\
                                                              A^*         & C         \end{array} \right]$
satisfies $D' = [B,D]$ using
$B = d-d^*= \left[ \begin{array}{cc} 0                     & A                    \\
                                    -A^*                   & 0          \end{array} \right]$ can be 
written as the matrix differential equation
\begin{equation}
B' = 2 A A^*, A' = 2 A C,  C' = -2 A^* A \; . 
\label{2}
\end{equation}
Since the quantity $AC+BA$ is time invariant,
$L$ is block diagonal with entries $B^2 + A A^*$ and $C^2 + A^* A$.
In Fourier space, $A,B,C$ are double infinite matrices. We have $B(0)=C(0)=A(\infty)=0$ and
$A(0) = {\rm Diag}(\dots -3i,-2i,-i,0,i,2i,3i, \dots)$ and
$B(\infty) = -C(\infty) = {\rm Diag}( \dots 3,2,1,0,1,2,3, \dots )$.
System (\ref{2}) shows initial inflation and asymptotic exponential expansion of the individual circles.
Also for general initial conditions $A,B,C \in M(n,C)$, system~(\ref{2})
satisfies $A(t) \to 0$. 

\begin{figure}
\scalebox{0.20}{\includegraphics{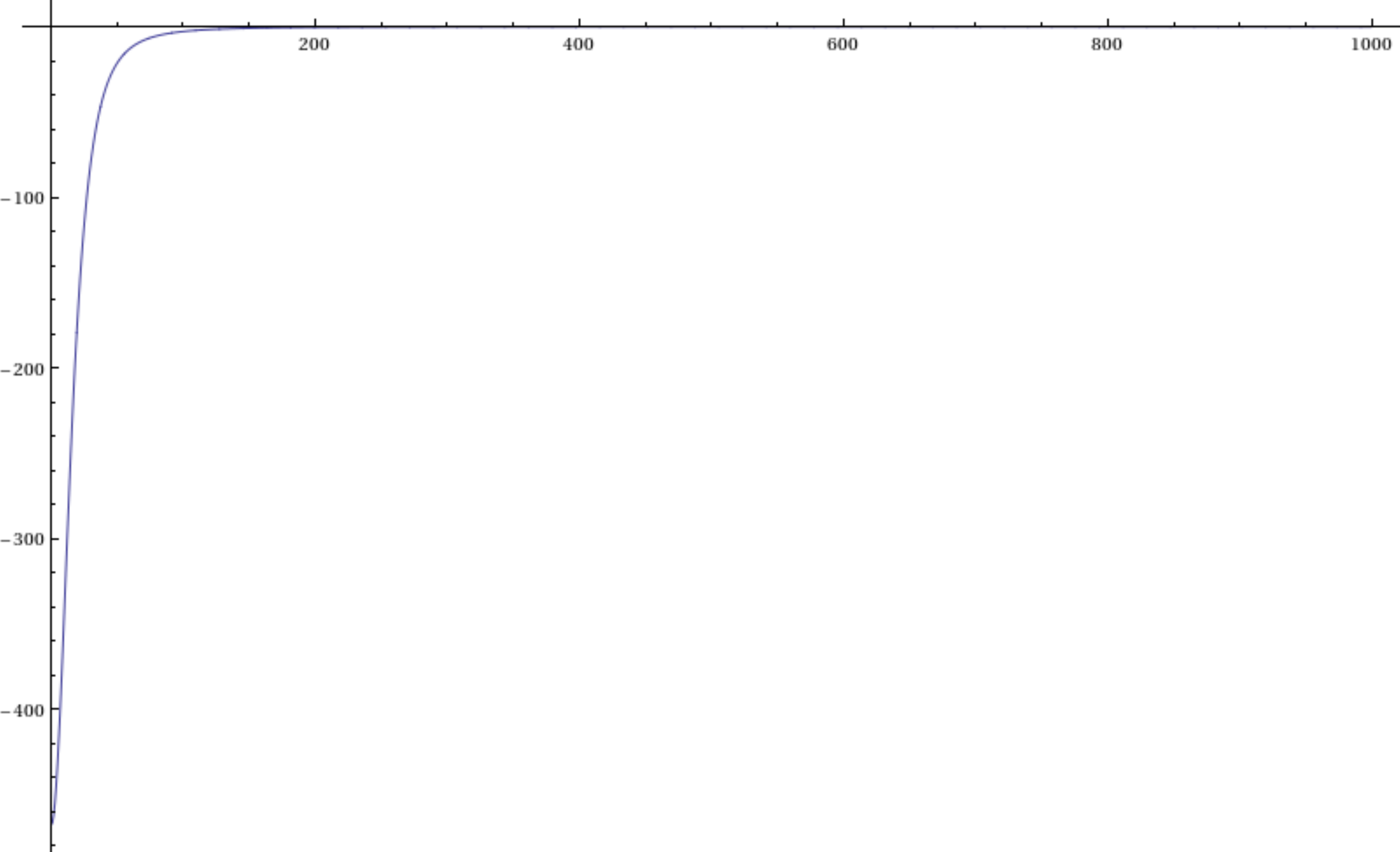}}
\scalebox{0.20}{\includegraphics{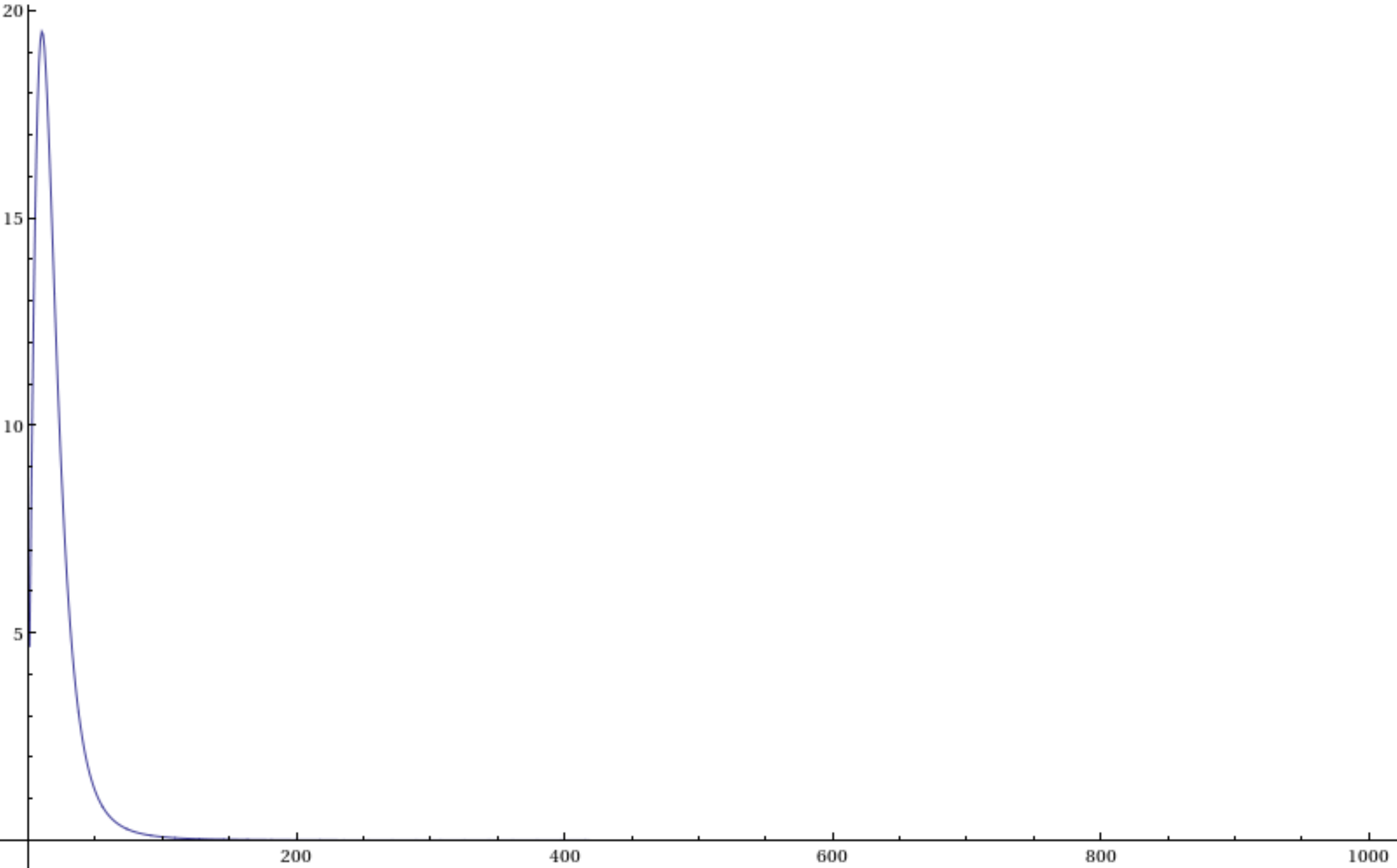}}
\caption{
Solutions to the system $D'=[B,D]$ for a one-dimensional compact Riemannian manifold, the circle. 
We see the evolution in the ${\rm tr}(A(t))$ and ${\rm tr}(A'(t))$. As in the graph case, we see an inflationary
expansion at first. The spectrum of the new Dirac operator $C(t) = d(t) + d(t)^*$ is not compatible
with the geometry of a one-dimensional Riemannian manifold. Actually, we do start with a union of two circle,
the space of $0$-forms and the space of $1$ forms. Initially, when using $D$ as the measuring stick, the two
worlds are separated: the distance between points $x$ and $Dx$ is infinite. The Connes pseudo metric still
gives infinite distance for positive $t$ because there is a kernel of $L_1$. But if we look at a distance
using to measuring the time a wave needs to go from one point to an other, then for 
positive $t$, we must think of the geometry as two circles separated by a small distance.
The distance between the planes becomes exponentially small and
internal distance in the planes expands. In the limit $t \to \infty$,
each of the two subspaces has adapted a discrete topology, in which the distance between two points 
is infinite.}
\end{figure}

{\bf Remarks.}  \\
{\bf 1)} The fact that all the integers and not only the positive integers appear as the 
spectrum of $D$ show that $D$ is an even more natural object than the Laplacian and the
Minakshisundaraman-Pleijel zeta function of the circle. The Dirac zeta function
$\zeta(s) = \sum_{n \neq 0} n^{-s} = \zeta(s) + (-1)^{-s} \zeta(s) = 
\zeta(s) [ 1+cos(\pi s) + i \sin(\pi s)]$ is now analytic in the entire complex plane because the pole $1$ has
been absorbed by $1+e^{-i \pi s}$. Of course, the Riemann zeta function and the zeta function of the circle
have the same roots $\zeta(z_i)=0$. Since $\zeta'_T(0) = -1$, the Ray-Singer regularized determinant of $D$ is $e$ 
and the Laplacian $D^2$ has the determinant $e^2$. 
The Zeta function of the Dirac operator on the circle is natural because it is 
equivalent to the Riemann zeta function but still analytic in the entire plane.  \\
{\bf 2)} For the two torus, the Dirac zeta function is
$\zeta(s) = (1+\exp(-i \pi s) \sum_{(n,m) \neq (0,0)} (n^2+m^2)^{-s/2}$. For odd dimensional tori, there is a branch 
such that $\zeta(s) = (1+\exp(-i \pi s)) \sum_{(n_1,\dots, n_d) \neq 0} (\sum_i n_i^2)^{-s/2}$ is analytic everywhere. 
The reason is that in general for odd dimensional compact Riemannian manifolds, 
the spectral zeta function has a meromorphic extension \cite{LapidusFrankenhuijsen}
to the complex plane with poles located at the odd integers $d-2q,q=0,1,2,...$. When going from the zeta function
for the Laplacian to the zeta function of the Dirac operator, the inclusion of the negative eigenvalues
produces the factor $(1+(-1)^s) = (1+\exp(-i \pi s))$ which kills the odd poles. In the even case, where the poles are at
even integers $s=d,d-2,d-4,...,4,2$, the factor $(1+\exp(-i \pi s))$ does not cover the poles any more. 

The deformed operator will be of the form
$$ D=D(t)= d+d^* + b = \left[ \begin{array}{cc} B                        & A                    \\
                                                A^*                      & C         \end{array} \right]  = 
\left[ \begin{array}{cc} 0                        & A                    \\
                         0                        & 0         \end{array} \right] + 
\left[ \begin{array}{cc} 0                        & 0                    \\
                         A^*                      & 0         \end{array} \right] +
\left[ \begin{array}{cc} B                        & 0                    \\
                         0                        & C         \end{array} \right]  \; . $$

which satisfies $D' = [B,D]$ using 
$$ B = \left[ \begin{array}{cc} 0                      & A                    \\
                                -A^*                   & 0          \end{array} \right] \; . $$
While initially $B=0$ and $A$ is a first order differential operator, this does not stay so 
after the deformation. The differential equations are:
\begin{eqnarray*}
B' &=& 2 A A^* \\
A' &=& 2 A C \\
C' &=& -2 A^* A
\end{eqnarray*}

\begin{lemma}
The quantity $AC+BA$ is left invariant under the motion.
\end{lemma}
\begin{proof}
The reason is that $A'=AC+BA=0$ is left invariant
$(B A + A C)' = B' A + B A' + A' C + A C' =  2 A A^* A + B (AC-BA) + (AC-BA) C - 2 A A^* A
                                         =  B (AC-BA) + (AC-BA) C = A C^2 -B^2 A$. And if
we take the derivative of this we get $A C^3+B^3 A$. 
\end{proof}

It follows that $L$ is diagonal with entries $B^2 + A A^*$ and $C^2 + A^* A$. 

Lets look at the evolution of the system in the Fourier space, where $A,B,C$
are double infinite matrices. We have $B(0)=C(0)=0$ and 
$$ A(0) = \left[ \begin{array}{cccccccccc}
                            ... & ...& ...  &   ...  & ...  & ...  & ...  & ... & ... & ... \\
                            ... &  0 &$-3i$ & 0      & 0    & 0    & 0    & 0   & 0   & ... \\
                            ... &  0 & 0    & $-2i$  & 0    & 0    & 0    & 0   & 0   & ... \\
                            ... &  0 & 0    &  0     & $-i$ & 0    & 0    & 0   & 0   & ... \\
                            ... &  0 & 0    &  0     & 0    & 0    & 0    & 0   & 0   & ... \\
                            ... &  0 & 0    &  0     & 0    & 0    & $i$  & 0   & 0   & ... \\
                            ... &  0 & 0    &  0     & 0    & 0    & 0    & $2i$& 0   & ... \\
                            ... & ...& ...  &   ...  & ...  & ...  & ...  & ... & ... & ... \\
                  \end{array} \right] \; . $$
In the limit, we get $A(\infty)=0$ and 
$$ B(\infty) = -C(\infty) = \left[ \begin{array}{cccccccccc}
                            ... & ...& ...  &   ...  & ...  & ...  & ...  & ... & ... & ... \\
                            ... &  0 & 3    &  0     & 0    & 0    & 0    & 0   & 0   & ... \\
                            ... &  0 & 0    &  2     & 0    & 0    & 0    & 0   & 0   & ... \\
                            ... &  0 & 0    &  0     & 1    & 0    & 0    & 0   & 0   & ... \\
                            ... &  0 & 0    &  0     & 0    & 0    & 0    & 0   & 0   & ... \\
                            ... &  0 & 0    &  0     & 0    & 0    & 1    & 0   & 0   & ... \\
                            ... &  0 & 0    &  0     & 0    & 0    & 0    & 2   & 0   & ... \\
                            ... & ...& ...  &   ...  & ...  & ...  & ...  & ... & ... & ... \\
                  \end{array} \right] \; . $$
We see that also in the limit, $t \to \infty$, we do not have a differential operator. While the 
superpartner of a $0$-form had been a $1$-form initially, the superpartner of a $0$-form is now
also a $0$-form. In the same way, a fermionic $1$ form has in the end a fermionic super partner. 

\section{Example: two point graph} 

The simplest graph in which a nonzero distance appears is the two point graph $K_2$. Because $v_0=2,v_1=1$, the
Dirac operator $D$ is the $3 \times 3$ matrix 
$$ \left[ \begin{array}{ccc} 0 & 0 & -1 \\ 0 & 0 & 1 \\ -1 & 1 & 0 \end{array} \right]   \; . $$
It has eigenvalues $-\sqrt{2},\sqrt{2},0$ and the Laplacian
$$ L= D^2= \left[ \begin{array}{ccc} 1 & -1 & 0 \\ -1 & 1 & 0 \\ 0 & 0 & 2 \\ \end{array} \right] \; . $$

When we run the differential equation for $\beta=0$, we see
$$ D(0) =  \left[ \begin{array}{ccc}  0 &  0 & -1 \\
                                      0 &  0 &  1 \\
                                     -1 &  1 &  0 \end{array} \right], 
   D(1) = \left[ \begin{array}{ccc}
                      0.702191  & -0.702191 & -0.117712 \\
                     -0.702191  & 0.702191  &  0.117712 \\
                     -0.117712  & 0.117712  & -1.40438 \end{array} \right] \;  $$
which then converges in the limit $t \to \infty$ to the projection-dilation 
$V^+ =  \left[ \begin{array}{ccc}
                      1         & -  1      &  0   \\
                     -1         &    1      &  0   \\
                      0         &  0        &  -2   \end{array} \right]/\sqrt(2)$.
Backwards in time, we get the limit
$V^- =  \left[ \begin{array}{ccc}
                     -1         &    1      &  0   \\
                      1         &   -1      &  0   \\
                      0         &  0        &   2   \end{array} \right]/\sqrt(2)$.
By symmetry, the differential equation is
$$ D = \left[ \begin{array}{ccc}
                      b         & c         & -d        \\
                      c         & b         &  d        \\
                     -d         & d         &  e       \end{array} \right] \; ,  
   B = \left[ \begin{array}{ccc}
                      0         & 0         & -d        \\
                      0         & 0         &  d        \\
                      d         & -d        &  0       \end{array} \right] \; ,  $$
which gives
$$ B D - D B = \left[ \begin{array}{ccc}
                  2 d^2 & -2 d^2 & b d-c d-e d \\
                  -2 d^2 & 2 d^2 & -b d+c d+e d \\
                  b d-c d-e d & -b d+c d+e d & -4 d^2
                 \end{array} \right] $$
and shows that $c=-b$ and $e=-2b$ (which also follows from the trace being zero) so that
$ D' = \left[ \begin{array}{ccc}
                   2 d^2 & -2 d^2 & 4 b d \\
                   -2 d^2 & 2 d^2 & -4 b d \\
                   4 b d & -4 b d & -4 d^2
                  \end{array} \right]$ and so that the differential equations are
\begin{eqnarray*}
   b'&=&2d^2 \\
   d'&=&-4bd \\
\end{eqnarray*} 
This system of nonlinear equations has the explicit solutions
$$ (d(t),b(t))=(\sqrt{1-\tanh^2(\sqrt{8} t)}, \tanh(\sqrt{8} t)/\sqrt{2}) $$
and the integral $d^2+2b^2=1$ which consists of ellipses. The derivative 
$$ d'(t)=-4 \sqrt{\frac{1}{\cosh \left(4 \sqrt{2} t\right)+1}} \tanh \left(\sqrt{8} t\right) \; $$
shows the inflation which is present in general. 

\begin{figure}
\scalebox{0.20}{\includegraphics{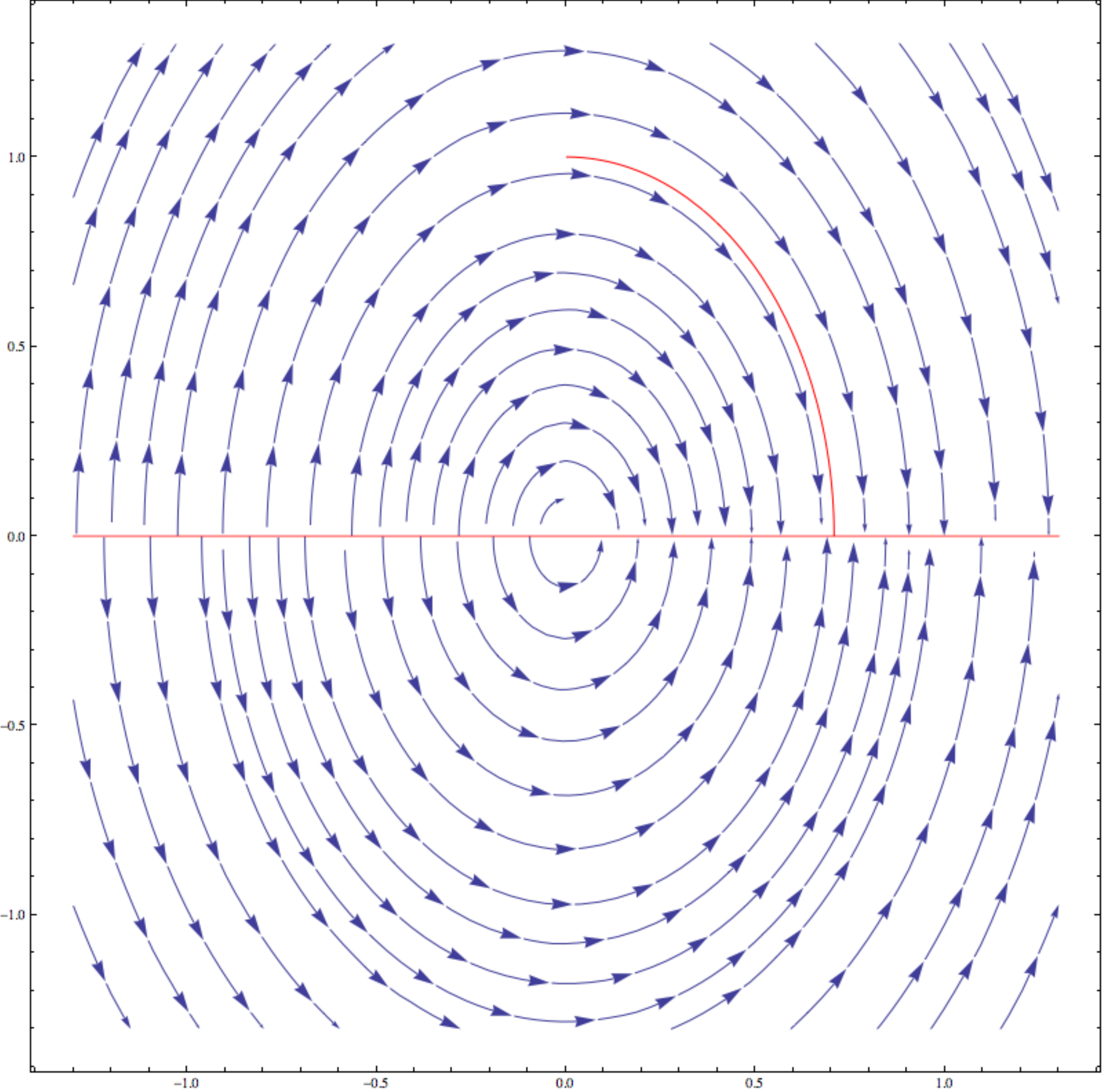}}
\scalebox{0.20}{\includegraphics{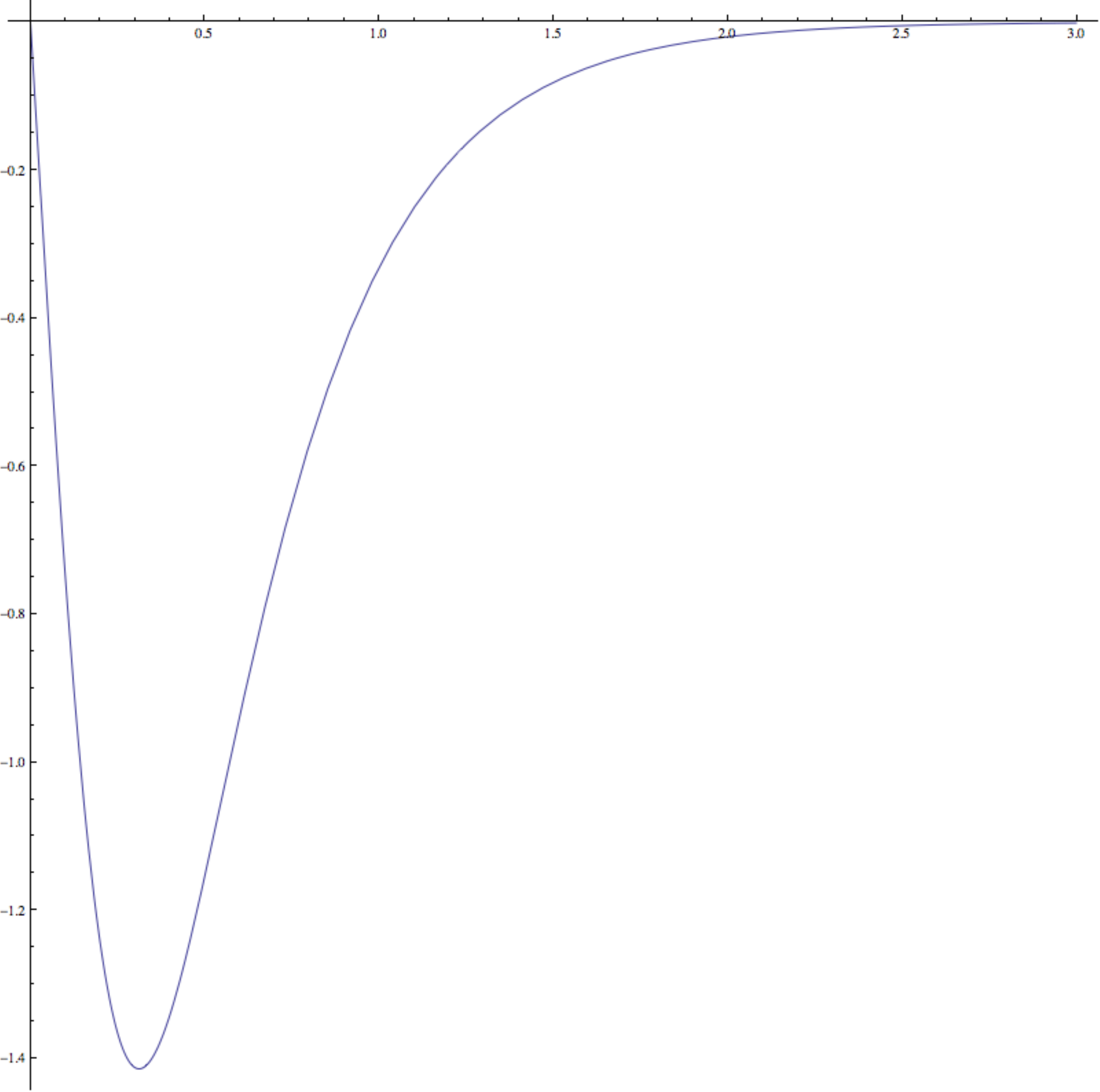}}
\caption{
The differential equation in the graph case $G=K_2$. We see the evolution in the $d(t),b(t)$
plane to the left and the function $d'(t)$ to the  right. The function $d'(t)$ has an extremum $-1/\sqrt{2}$ 
at $t={\rm arccosh}(\sqrt{(1+\sqrt{2})/2})/\sqrt{2} = 0.311613...$ which pinpoints
the inflection point for $d(t)$. Unlike the inflationary expansion of the universe
$10^{-36}$ to $10^{-33}$ seconds after the big bang, which expands the universe by a factor 
$10^{78}$, the expansion is only $\sqrt{2}$ which is the order of magnitude for the 
expansion rate for a general graph. To get larger expansion rates, the initial 
operator $D = \sum_i d_i + d_i^*$ has to be replaced by $D=\sum_i c_i (d_i + d_i^*)$. }
\end{figure}

The inflation rate is pretty much independent of the graph. To get bigger inflation rates
we have to scale the different exterior derivatives $d_k: \Omega_k \to \Omega_{k+1}$ 
differently. This corresponds to choosing units on each p-form sector. 
Lets look at the triangle graph $K_3$ for which the Dirac operator is a $7 \times 7$ matrix. 
Since we start with the neutral Dirac operator $D$ which has symmetry, we only have $6$ variables
$b_i$ to describe $b$ and $2$ variables $d_i$ to describe $C$.  Lets leave $d_0$ as it is and change 
$d_1$ by a factor $10$. This is natural when looking at the 0-form forms and 1-forms as
``branes" in the total space. The integrable differential equations $D'=[B,D]$ for the Dirac operator
$$ D= \left[
   \begin{array}{ccccccc}
    b_1 & b_2  & b_3 & d_1  & d_1 &   0  & 0 \\
    b_2 & b_1  & b_2 & -d_1 & 0   & d_1  & 0 \\
    b_2 & b_2  & b_1 & 0    & -d_1& -d_1 & 0 \\
    d_1 & -d_1 & 0   & b_4  &  b_5&  0   & -d_2\\
    d_1 & 0    & -d_1& b_5  &  b_4&  b_5 &  d_2 \\
    0   & d_1  & -d_1&  0   &  b_5&  b_4 & -d_2 \\
    0   & 0    & 0   & -d_2 &  d_2& -d_2 &  b_6
   \end{array}
   \right]  \; . $$
of the triangle $G$ with symmetry are the nonlinear system of equations
\begin{eqnarray*}
   b_1' &=& 2 2 b_2^2+2b_3^2+4 d_1^2 \\
   b_2' &=& 2b_2 b_3 -2d_1^2 \\
   b_3' &=&-2d_1^2 \\
   b_4' &=& 2b_5^2-4d_1^2+2d_2^2 \\
   b_5' &=& -2d_1^2-2d_2^2   \\
   d_1' &=& (2b_5-b_1+b_4)d_1  \\
   d_2' &=& (-b_4+2b_5+b_6)d_2  \\
\end{eqnarray*}
It is already too complicated for computer algebra systems to 
find explicit analytic solutions. We know however that it is integrable
and that $d_i \to 0$. 

\begin{figure}
\scalebox{0.12}{\includegraphics{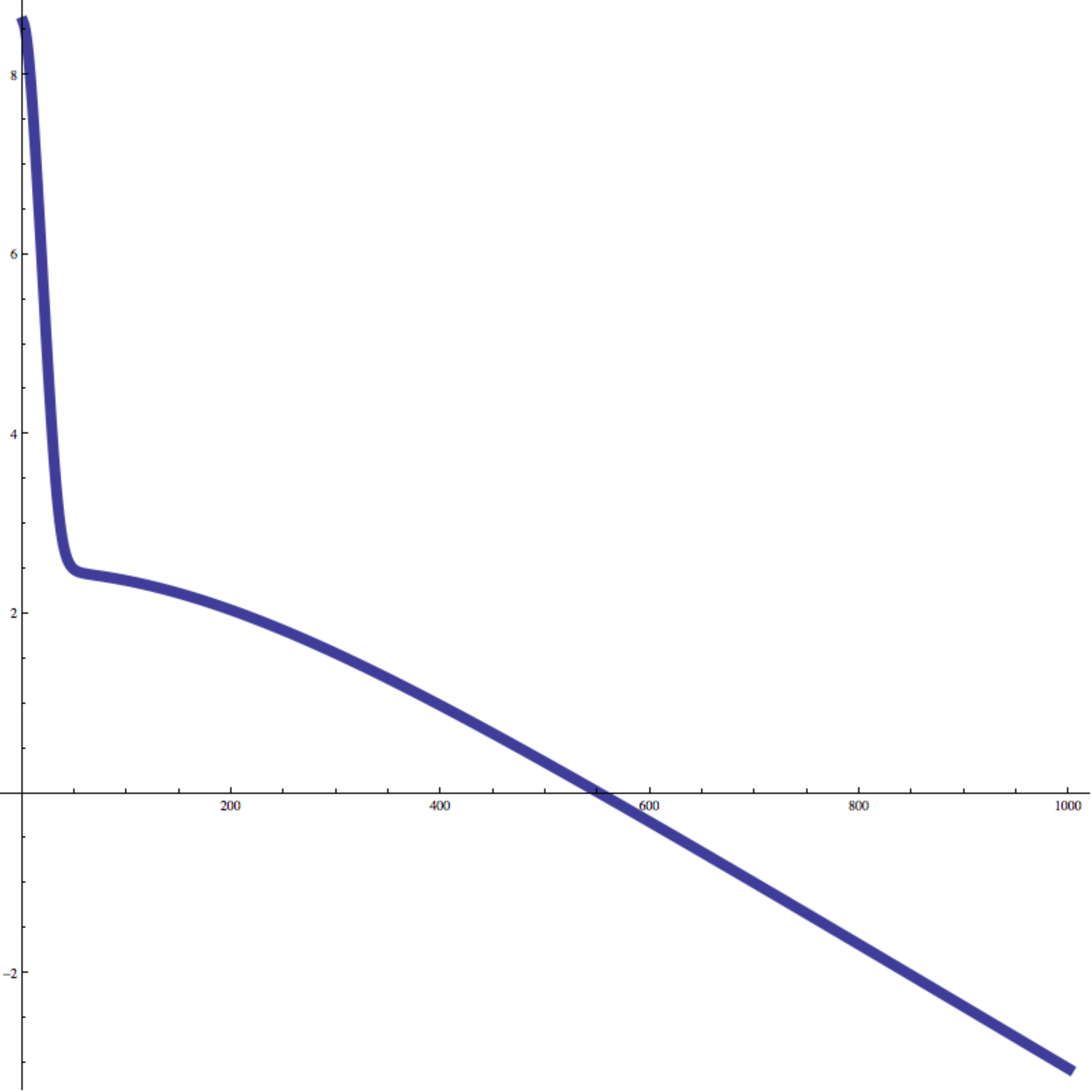}}
\scalebox{0.12}{\includegraphics{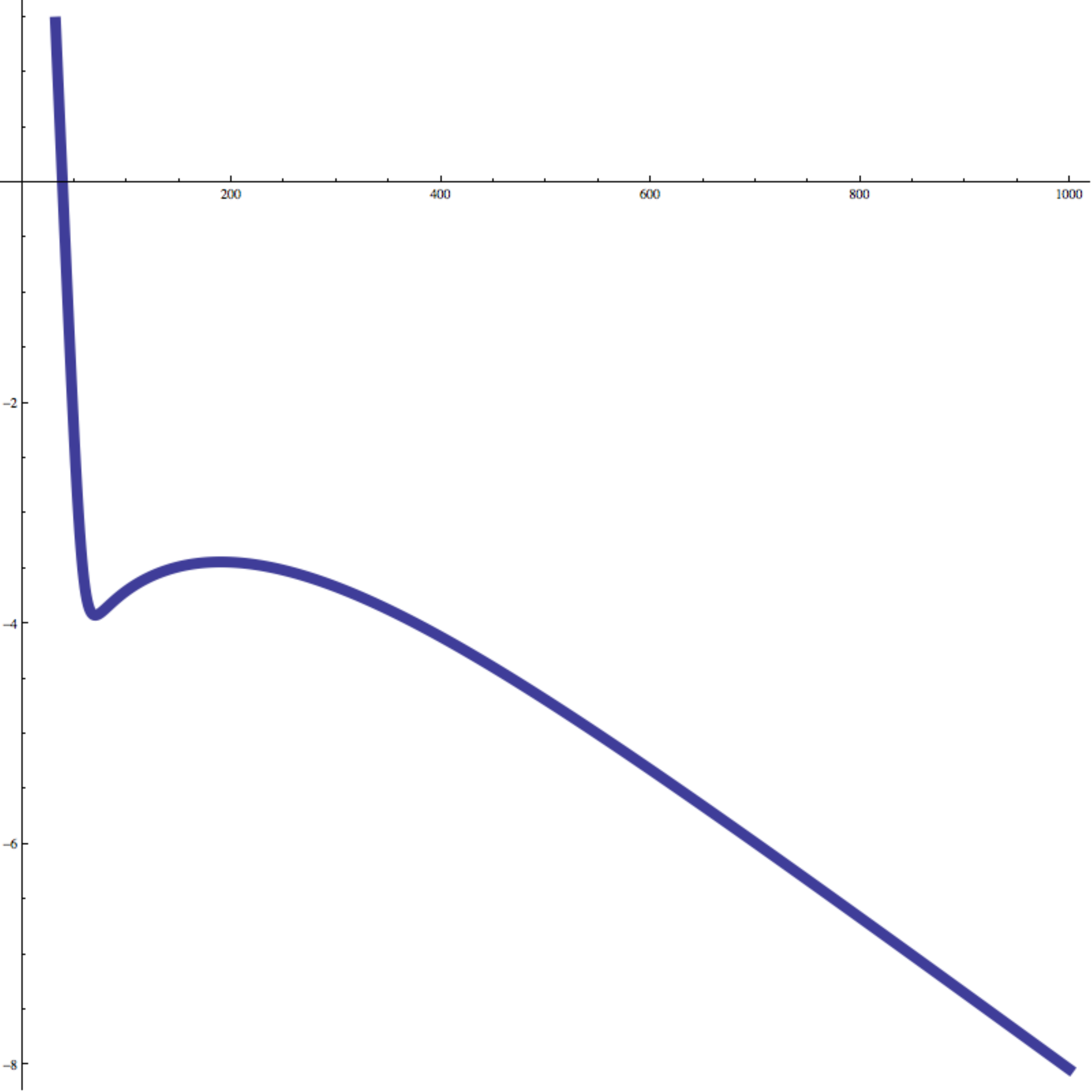}}
\caption{
The graph of the decaying function ${\rm tr}(M(t))$ as well as its derivative
$\frac{d}{dt} {\rm tr}(M(t))$ in logarithmic coordinates in the case of the triangle $G=K_3$. This is the
smallest case, where it is already possible to tune the initial condition for the Dirac operator
by adding coupling constants to $d_k: \Omega_k \to \Omega_{k+1}$. We have
scaled $d_1$ by a factor $k=10$. We see now extra ordinary large inflation.
}
\end{figure}

\section*{Appendix: The zeta function}

\begin{figure}
\scalebox{0.12}{\includegraphics{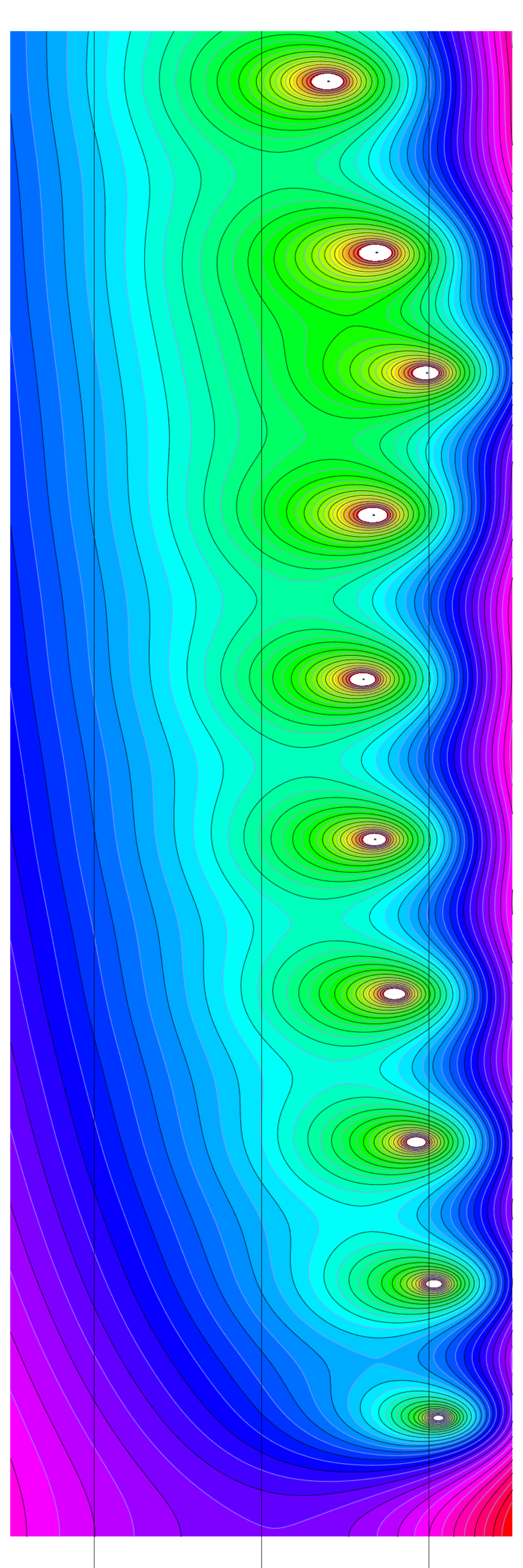}}
\scalebox{0.12}{\includegraphics{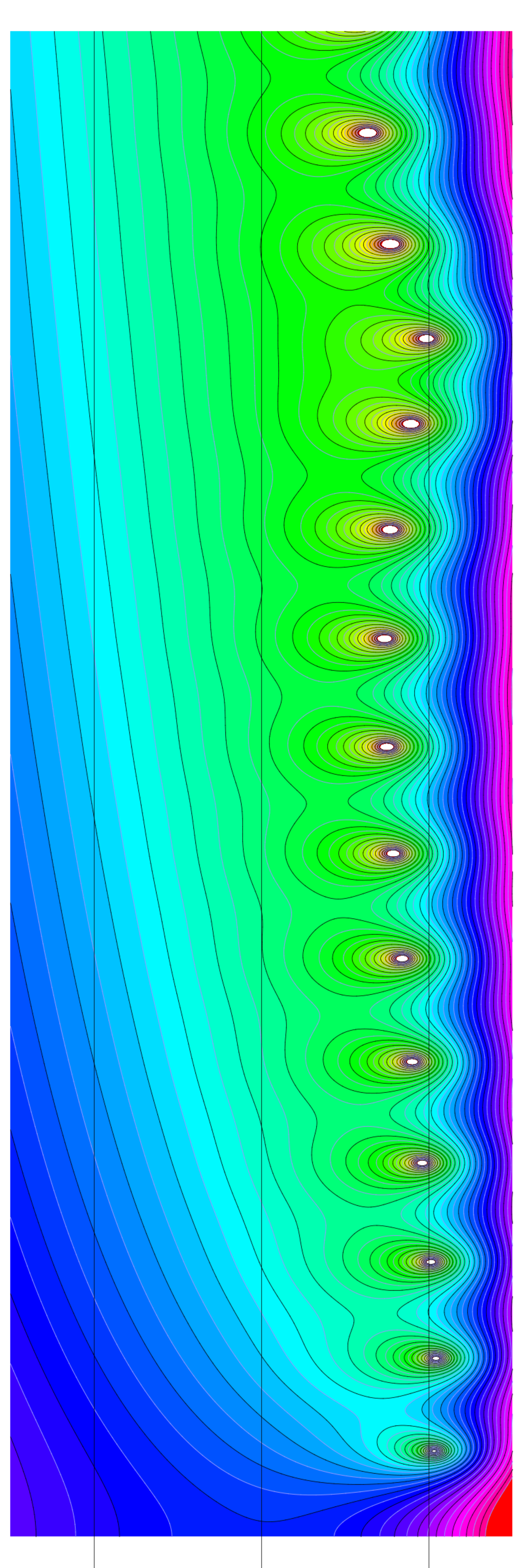}}
\scalebox{0.12}{\includegraphics{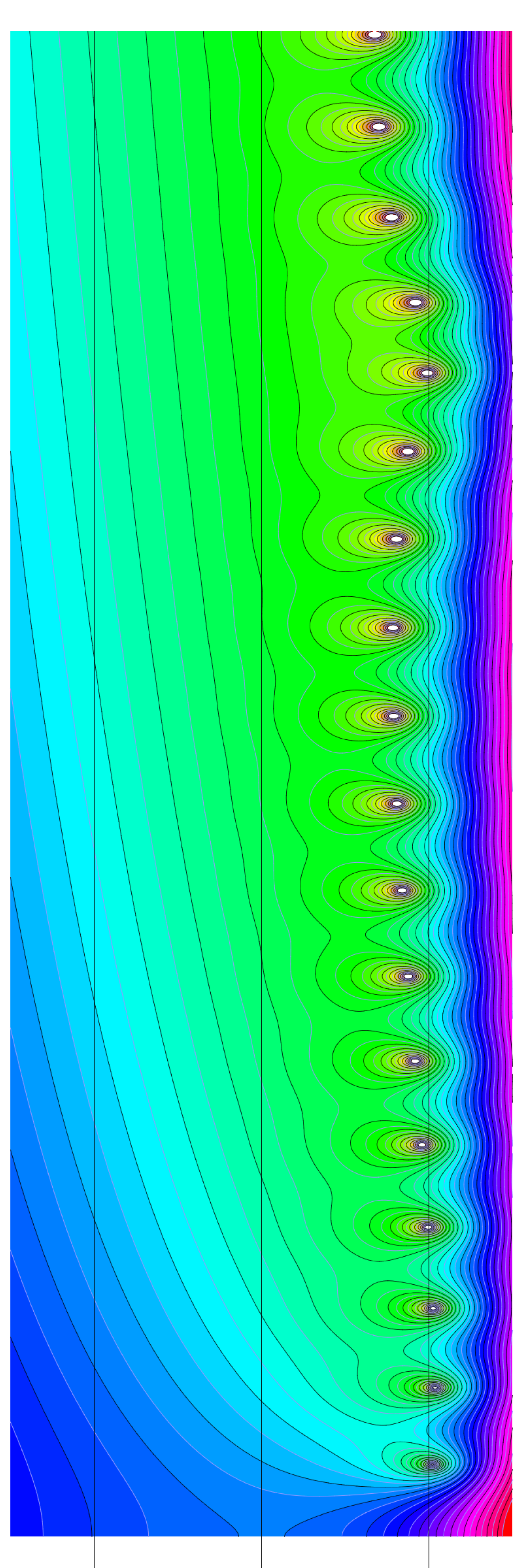}}
\caption{
The zeta function of the Dirac operator of the circular graphs $C_{10}$ (left), $C_{500}$ (middle)
and $C_{1500}$ (right) for $s = a + i b \in [-1.5,1.5] \times [0,18]$. The vertical division lines are
drawn at $b=-1,0,1$. 
We see the roots of the analytic function $\zeta(s) = (1+e^{-i \pi s}) \sum_{k=1}^{n-1} \sin^s(\pi k/n)$. Since
we are interested in the roots only and the factor $(1+e^{-i \pi s})$ grows exponentially on
the negative imaginary axes, we ignored anti-matter (negative eigenvalues) and 
just plotted the level curves of $\zeta_0(s) = \sum_{k=1}^{n-1} \sin^s(\pi k/n)$
instead which has the same roots. It would be interesting to know what the location of the roots are 
in the limit $n \to \infty$. This zeta function is the discrete analogue of the Dirac zeta function 
of the circle which is the analytic $\zeta(s) = (1+e^{-i \pi s}) \sum_{k=1}^{\infty} k^{-s}$ and which 
has the same roots than the Riemann zeta function $\zeta_0(s) =  \sum_{k=1}^{\infty} k^{-s}$. 
\label{inflation}
}
\end{figure}

The Dirac zeta function of a Dirac operator $D$ is defined as 
$$   \zeta(s) = \sum_{\lambda \neq 0} \lambda^{-s}  \; , $$ 
where the sum is over all nonzero eigenvalues of $D$. Since $\lambda^2$ are the eigenvalues of $L$.
But because taking roots $(-\lambda)^s$ can be done in different ways, and $\lambda$ can be negative, the function
$\zeta(s)$ needs to be specified more precisely. We will chose a branch by relating $\zeta(s)$
with $\zeta_{MS}(s)$, the Minakshisundaram-Pleijel zeta function $\zeta_{MS}$: define
$$ \zeta(s) = \zeta_{MP}(s/2) + (-1)^s \zeta_{MP}(s/2) = \zeta_{MP}(s/2) [ 1+e^{-i \pi s}]  \; . $$

\begin{lemma}
The Dirac zeta function $\zeta(s)$ has a meromorphic continuation to 
the entire complex plane. For odd dimensional manifolds, there is an analytic continuation
to the entire complex plane. 
\end{lemma}

\begin{proof} 
In the graph case, the function $\zeta_{MP}$ is a finite sum and already analytic. 
In the manifold case, the function $\zeta_{MP}$ has a meromorphic extension to 
the entire complex plane \cite{LapidusFrankenhuijsen}. There a simple poles located
at $s=d$ and a subset of the points $d-2,d-4, \dots$, where $d$ is the dimension of $M$.  
In odd dimensions, there are simple poles at $d,d-2,d-4, \dots,... $ and in even 
dimensions simple poles at $d=,d-2,d-4,\dots, 4,2$. 
The factor $1+e^{-i \pi s}$ which has roots at $s=1,3,5,...$ 
regularizes the poles of the Laplace zeta function $\zeta_{MS}$ in the odd-dimensional 
case. 
\end{proof} 

\begin{coro}
For odd dimensional manifolds, $\tr(D^n)$ can be defined for all $n \in \Z$
For even dimensional manifold $\tr(D^n)$ is defined for all $n \geq 0$. 
\end{coro}

{\bf Remarks.} \\
{\bf 1)} For the circle, the Dirac zeta function has the same roots then the classical 
Riemann zeta function. \\
{\bf 2)} Fix a manifold $M$. For every Riemannian metric on $M$, we have a Dirac operator $D$.
By analytic continuation, the trace is defined for polynomials of these operators. If one could
extend it to squares $\tr((A+B)^2)$, one could define $\tr(AB+BA)$ and use it as 
an inner product on $\Sigma$.  \\
{\bf 3)} The Laplacian $L$ on a graph or manifold is related to various differential equations
the  heat equation $\dot{u} = L u$, the wave equation $\ddot{u}=-Lu$, 
the Maxwell equation $dF=0,d^*F=j$ relating electromagnetic field $F$ with matter $j$, 
which is in a Coulomb gauge $d^*A=0$ for the vector potential $A$ equivalent to the 
Poisson equation $LA=j$, where $dA=F$ which is in the vacuum case $j=0$ the Dirichlet problem 
$LA=0$ with harmonic solutions $A$, the Schr\"odinger equation 
$\dot{u} = i L u$, the Dirac wave equation $\dot{u} = i D u$. \\
{\bf 4)} The Dirac wave equation $\dot{u} = i D u$ looks like the Schr\"odinger equation but is also closely tied to the wave equation
as d'Alembert has shown: we can factor $(\partial_t^2+D^2)u=0$ and get $(\partial_t - i D) u=0$
or $(\partial_t + i D) u = 0$. This leads to the explicit solutions 
$u(t) = \cos(t D) u(0) + \sin(t D) D^{-1} u'(0)$. The assumption that the initial velocity $u'(0)$
is perpendicular to the kernel of $D$ is natural as it is in one dimensions, where solutions
to the wave equation $u_{tt} = u_{xx}$ assume that the center of mass of $u'(0)$ is zero. If
it is not zero, we could change the coordinate system in order to have the string motion $u(t,x)$
not translate freely in space. \\
{\bf 5)} We see that the Dirac operator allows to treat the wave equation on a general finite simple graph or
Riemannian manifold in the same way than the wave equation is dealt with in one dimensions, 
where $D=id/dt$ and $e^{i t D} u = e^{-d/dt t}  u (x) = u(x-t)$ by Taylor's theorem.  \\
{\bf 6)} The unitary Dirac evolution $e^{i t D}$ is sometimes
called the wave group (i.e. \cite{Zelditch}), but $D$ is usually define
as a pseudo differential operator defined by the spectral theorem. The
Dirac $D$ produces a natural square root, but it needs an extension of $L$ to all differential
forms. \\
{\bf 7)} Also nonlinear integrable systems like the sine-Gordon equation $u'' - L u = \sin(u)$ 
can be considered also at on differential forms or on graph, even so it is likely 
that $\sin$ has to be modified in the discrete to preserve integrability. 

\section*{Appendix: integrability}

Integrability has many meanings. Informally it implies that a system has enough conserved quantities so that
it becomes solvable. For ordinary or partial differential equations, it usually means to explicit
solution formulas can be written down; this happens often algebraically but the later is not a necessity:
the quadratic system $t x'=(x^2-y)/6,t y'=(xy-z)/3,t z'=(xz-y^2)/2$ of 
Ramanujan \cite{Ward2003} for example is explicitly solved by 
$x(t)=1-24\sum_{n=1}^{\infty} \sigma_1(n) t^n,y(t)=1+240 \sum_{n=1}^{\infty} \sigma_3(n) t^n,
z(t)=1-504 \sum_{n=1}^{\infty} \sigma_5(n) t^n$ using number theoretic $\sigma$ functions
$\sigma_k(n)$ which is not algebro-geometric. The system still must be considered integrable.  \\
The question "what is integrability" is often discussed informally \cite{Knill2002,Veselov2008,overflowintegrable}. 
In the overwhelming cases where the group $Z$ or $R$ acts on a topological space $X$,
integrable systems have the property that every
trajectory converges in a compactified phase space $\overline{X}$ both forward or backwards to a group 
translation on a compact topological group. This means that for every invariant measure in a compactification,
the induced Koopman operator has pure discrete spectrum. Informally, a system is not integrable if it shows
any kind of chaos, this can be weak chaos like invariant measures with weak mixing, stronger chaos like
mixing invariant measures or strong chaos like invariant measures on which the system is 
a Markov process or equivalently has a Bernoulli shift as a factor.
Positive topological entropy or equivalently positive Kolmogorov-Sinai entropy for some invariant measure also 
prevents a system from being integrable. \cite{Katok2007} gives a historical overview of entropy.
A system is not integrable if there is an invariant measure for which the induced system has anything 
else than pure point spectrum. An invariant horse-shoe for example induced to the existence of a 
transversal homoclinic point and this is enough to 
kill integrability. Integrability can justify the introduction of new functions.
The pendulum $x''=k \sin(x)$ for example can be solved using elliptic functions:
since energy $E=\dot{x}^2-k \cos(x)$ is conserved, we have $\dot{x} = \sqrt{E+k\cos(x)}$ so that $x(t)$
is implicitly given by $t=\int dx/\sqrt{E+k\cos(x)}$ which is an elliptic integral of the first kind.
Consequently, $x(t)$ is the inverse of an elliptic integral, an elliptic function. 
Any Hamiltonian system of one degree of freedom is integrable because the energy surfaces is one dimensional
and every trajectory either is located on a closed loop, is on a path to infinity or converges to a fixed point. 
More generally, any ordinary differential equation in two dimensions is integrable 
by the theorem of Poincar\'e-Bendixon \cite{Perko}: the uniqueness of solutions to differential equations prevents
paths to cross and both forward or backwards, orbits approach circular or point attractors or escape to 
infinity. Examples like the ABC flow $x'=A \in(z) + C \cos(y), y'=B \sin(x) + A \cos(z), z'=C \sin(y) + B \cos(x)$,
the Lorentz system $x'=10 (y - x), y'-x z + 28 x - y, z'x y -8/3 z$, the Roessler systems or periodically driven penduli
show that in three dimensions already, integrability fails in general. Simple maps in the plane like the 
H\'enon map $T(x,y)=(x^2-c-y,x)$ or the Chirikov Standard map $T(x,y)=(2x+c\sin(x),x)$ on the torus $R^2/(2\pi Z)^2$ 
show that for discrete time, area preserving maps can be non integrable already. Already one-dimensional interval maps like the 
Ulam map $T(x) = 4x(1-x)$ which is conjugated to the piecewise linear tent map $S(x) = 1-2 |x-1/2|$ with the conjugation 
$U(x)=\frac{1}{2}-\frac{1}{\pi} {\rm arcsin}(1-2x)$. The Feigenbaum family $T_c(x) = c x(1-x)$ shows that integrability
and chaos can be woven together in a complicated way, depending in an intriguing way on the parameter $c$, like that 
chaos is approached using period doubling bifurcations. Naive discretizations can destroy integrability. 
It was not obvious, how to find the right formula for an integrable discretisation of the pendulum from the continuous version $\ddot{x} = \sin(x)$.
It is given by the Suris-Bobenko-Kutz-Pinkall map $T(x,y) = (2x+4 \arg(1+k e^{-ix})-y,x)$ on the two torus $T^2$ 
which has the integral $F(x,y) = 2 (\cos(x)+\cos(y)) + k \cos(x+y) + k^{-1} \cos(x-y)$ \cite{Sur89,Bo+93}. 
It is not the Chirikov standard map $x_{n+1} - 2 x_n + x_{n-1} = k \sin(x_n)$ but this SBKP map 
$x_{n+1} - 2 x_n + x_{n-1} = 4 \arg(1+k e^{-ix})$ which discretizes the pendulum $x''=k \sin(x)$ so that
integrability is preserved. Similarly, polynomial ordinary differential equations $x''=p(x)$, which are integrable in the continuum,
become complicated in the discrete and some insight is needed to 
find maps like the MacMillan family $T(x,y) = (\frac{2 k x}{(1+x^2)} - y, x)$ \cite{McM71} 
which has the integral $F(x,y) = x^2+y^2 + x^2 y^2 - 2k x y$. 
Also in numerical analysis, it is well known that perfectly well behaved integrable and solvable systems can become 
unstable after a naive numerical discretizations. Ingenious schemes have to be devised so that important features
of the continuum survive the discretisation. In general, one wants symmetries from the continuum to be inherited, have
global existence of solutions, and Hamiltonian structures preserved if preent. For finite-dimensional Hamiltonian systems, 
there is Liouville integrability, the existence of enough independent conserved quantities in involution.
A theorem of Liouville \cite{Arnold1980,Babelon2003} assures that the system can be integrated in that case.
A more general inductive notion is Frobenius integrability which means that the system admits a foliation
by integral manifolds on each of which the system is Frobenius integrable.
In noncompact cases, especially in scattering situations, the asymptotic velocities can be
integrals. The point at infinity could be included into a compactified phase space so that asymptotic quantities
at infinity are integrals.  An example is part of the phase space of the St\"ormer problem, a 
single particle in a magnetic dipole field. It is non-integrable in trapped parts of the phase 
space because of horse shoes \cite{Brown81} but it is integrable in other parts, where particles escape 
to infinity \cite{Mo63}. Coexistence of integrability and chaos has been known for a long time in the case of complex 
dynamics of a polynomial, where the map is integrable on the Fatou set
and chaotic on the Julia set \cite{Beardon}. For conservative systems, coexistence is more subtle \cite{Str89} but
not impossible \cite{Woj81}. A variant of asymptotic free motion is the inverse scattering approach which works for 
many integrable systems with non-compact phase space. Birkhoff integrability \cite{Birkhoff} is the notion
that the system can be linearized around periodic orbits and that the union of these linearized regions are
dense. Unlike Frobenius integrability, it still would allow for a set of zero measure, on 
which the system is complicated. The complex map $z \to z^2$ on the Riemann sphere is Birkhoff integrable
but not integrable because the map induced on $|z|=1$ is a Bernoulli shift. 
Frobenius integrability can apply in non-smooth situations, where analytic expressions are
impossible. For Hamiltonian systems on the cotangent bundle of a compact manifold, the 
Liouville-Arnold theorem assures that the system is conjugated to a linear flow on a torus. 
Liouville integrability can extend to infinite-dimensional situations like the KdV system
$u_t + u u_x + u_{xxx}$, Boussinesq equation $u_{tt} -u_{xx} = (u^2)_{xx} - u_{xxxx}$, 
Sine-Gordon $u_{tx} = \sin(u)$ or nonlinear Schr\"odinger equation
$iu_t = u_{xx} + |u|^2 u$. For integrable isospectral Toda deformations of random Schr\"odinger operators,
where the spectrum of the operators can be pretty arbitrary, the integration can be done by 
approximation with finite dimensional integrable systems \cite{Kni93a,Kni93b}. 
Both classical and quantum mechanics can be studied with operator methods. Classical systems use the 
Koopman transfer operators $Uf=f(T)$ or Perron-Frobenius transfer operators $Tf(x) = \sum_{y \in T^{-1}(x)} f(y)$. 
Quantum systems are described with the unitary evolution $U = \exp(i Lt)$. For classical systems, integrability 
means the operator $U$ has discrete spectrum. For quantum dynamics
$\dot{x} = i \hbar H x$, integrability could force that $H$  have discrete spectrum
under natural natural boundary condition. The dynamics itself on the unit ball of the Hilbert space is 
always integrable \cite{Kni97}. 
For a particle in a periodic potential, 
where $H$ has absolutely continuous band structure, we can consider a boundary condition at a point to get 
discrete auxiliary spectrum. A particle in a periodic potential is integrable because 
solutions are nice Bloch waves and imposing a zero boundary condition produces discrete auxiliary spectrum in the gaps
which the Abel map conjugates to a linear flow on a torus. Also for quantum mechanics there is not much agreement,
what systems should be called integrable: while the quantum mechanical harmonic oscillator $L=-\partial_x^2 + x^2-2$ should
definitely be called integrable because the eigenfunctions can be constructed recursively using the decomposition $L=A^*A=A A^*-2$ with 
$A=x+\partial_x,A^*=x-\partial_x$, one can argue whether  $-x''+V'(x)=0$ with polynomial $V$ should be 
called integrable. Still, there is a countable set of eigenvalues and eigenfunctions which explicitly solve any system 
like the heat $u_t=Lu$, wave $u_{tt} = L u$ or Schr\"odinger equation $iu_t =Lu$. An other notion for integrability 
in the quantum setting is that the classical limit is integrable. This is analogue to the notion of ``quantum chaos" which
sometimes is defined as the property of corresponding classical system is  ergodic \cite{KnaufSinai} or
Anosov. Related is quantum unique ergodicity, which is the property that for any observable $A$
and eigenfunctions $\phi_j$ belonging to eigenvalues $\lambda_j^2 \to \infty$ the
property $(A \phi_j,\phi_j) \to \int_{S^*M} \sigma_A \; d\mu$ where $\sigma_A$ is the 
principal symbol of $A$ and $\mu$ is the Liouville measure on the unit cotangent tangent bundle $S^*M$ 
of the manifold $M$ \cite{Zelditch}. The Dirichlet problem in a compact convex planar region $G$ would be integrable with
this notion if the corresponding billiard system were integrable. The Schr\"odinger or wave equation on a Riemannian 
manifold $M$ would be called integrable, if the geodesic flow on $M$ is integrable.
It can be understood in the sense that eigenvalues and eigenfunctions can be constructed 
explicitly like for the quantum harmonic oscillator. 
There is a quantum analogue of Liouville integrability in which Poisson commuting observables
are replaced by commuting observables. The notion is not unproblematic \cite{Fadeev}. 
For a free quantum mechanical particle on a manifold, quantum Liouville integrability implies Liouville 
integrability for the geodesic flow. 
In statistical mechanics, integrability can mean that asymptotic quantities have explicit expressions. 
For higher dimensional systems like tiling systems on a Lie group $G$, on which $G$ acts by translation, 
integrability can be defined as the fact that the unitary evolution has discrete spectrum. One can then define
a configuration $x$ to be a ``crystal" if the orbit through $x$ is the entire space. 
Quasicrystals \cite{Senechal} are systems for which space translation has discrete spectrum. 
This should be considered the case of ``integrable crystals". Crystals with singular continuous spectrum 
are called turbulent. They are called chaotic if the spectrum is absolutely continuous. 
DS-integrability was defined in \cite{Knill2002}.
For many integrable systems in a dynamical context, ordinary or partial differential equations in particular,
integrability manifests itself in the existence of a  Lax pairs $\dot{L}=[B,L]$ for some Lie algebra-valued 
operators $B,L$. There is often a geometric representation of integrable systems
as zero curvature equations $A_t-B_t+[A,B]=0$ for a connection in the sense that 
the covariant derivatives $D_x = \partial_x-A,D_y=\partial_t-B$ commute. One often sees also a 
biharmonic structure, two different Poisson brackets. Many systems - but not all - feature 
solitons, special localized solutions which do not change shape and interact with other solitons. 
Nonlinearity manifests itself that the amplitude of a wave has an influence on the speed of the wave. 
The Korteweg de Vriews system $u_t + u u_x + u_{xxx}=0$ for example has the solution 
$a \; {\rm sech}^2(b(x-ct))$ if $a=12b^2,c=4b$, where $b$ is a free parameter. These solutions were found
independently by Boussinesq and Rayleigh and show that the amplitude $a$ can determine the speed $c$
of the soliton. 
Many integrable systems also feature more symmetries. Examples are  B\"acklund transformations
which allow to construct new solutions from given ones. In the case of the Toda lattice
$L'=[B,L]$ for example, we can write $L=D^2+E$ for some constant $E$, where $D$ is an 
operator on a doubled lattice. If $\sigma$ is the shift on the later, we can write 
$L=S S^*+E$ where $S=D \sigma$. The B\"acklund transformed operator is $BT_E(L)=S^* S+E$. 
For ordinary differential equations with solutions $r(x,t)$, there is the Painlev\'e property,
which tells that the critical points $z \to r(x,z)$ do not depend on time. A conjecture of Ward 
states that any ordinary or partial differential equation which is 
``integrable" is obtained from a self-dual Yang-Mills gauge field
by reduction \cite{Ward2003}. Many inverse scattering problems can be related to a 
Riemann-Hilbert problem in complex variables. There are also ergodic theoretic connections: 
integrable systems are required to have zero topological entropy $h(T)$ \cite{AKM} 
so that for any invariant measure $\mu$, also the measure theoretic system has zero metric entropy by the 
Goodman inequality $h_{\mu}(T) \leq h(T)$ \cite{Goodman}.
All systems defined by a group $G$ acting on a topological space $X$ we know to be integrable have the property 
that for every invariant measure $\mu$ the dynamical systems $(X,T,\mu)$ has discrete spectrum. 
For many integrable systems $T(t)$, the time average 
$\phi(x) = \lim_{T \to \infty} (1/T) \int_0^T f(T_t(x)) \; dt$ converges to a continuous function
whose level surfaces foliate the phase space. A simple system which illustrates this is the Knuth map
$T(x,y) = (|x|-y,x)$ which is integrable because $T^9(x,y)=(x,y)$ so that with $\phi(x,y)=y$ the function 
$F(x,y) = \sum_{k=1}^9 \phi(T^k(x,y))$ is an integral.  
Many known integrable systems have their origin in physics. It starts with the Newtonian two body 
problem or equivalently, the Kepler problem. In the 19th century, other systems were added, like 
the Euler problem with two fixed attracting centers and the Vinti problem concerning the motion of a satellite
around an ellipsoid, an other special case of a generalization of the Euler problem \cite{Mathuna}. For 
rigid body motion, there is the free evolution of a compact solid in $n$ dimensions, and the Euler and Kovelevskaya tops.
For geodesic flows, the ellipsoid was solved by Jacobi. For surfaces of revolution, the Clairot integral 
renders the flow integrable.  Among nonlinear integrable partial differential equations, 
the Korteweg-de Vries equation was both experimentally and theoretically studied earliest.
At the beginning of the 20th century, it became clear that nonintegrability is generic. 
The discovery of solitons by Kruskal and Zabusky, experiments of Fermi-Pasta-Ulam as well as 
KAM perturbation results revived the subject and led to the study of completely integrable partial 
differential equations via inverse scattering methods introduced by Gardner, Green, Kruskal and Miura.  
Many systems, among them finite particle systems like 
the Toda system with Hamiltonian $\sum_i p_i^2 + \exp(q_i-q_{i+1})$ \cite{Gutzwiller}
or Calogero-Moser particle systems have been extended  and generalized, both with continuous and discrete time. For discrete time, 
the $QR$ system $L=QR \to RQ$ has gained much attention (i.e. \cite{Symes}) 
because it leads to a numerical method to 
diagonalize a matrix $L$. In periodic situations, where the motion is recurrent, the integration 
is done by algebra-geometric methods. Also integrable geometric evolution equations have been known for a long time:
the vortex filament flow $\dot{x}=x' \times x''$ introduced by Da Rios
at the beginning of the 20'th century can be reduced to the nonlinear Schr\"odinger equation
$iu_t = u_{xx} \pm  2 |u|^2 u$.  \cite{langerperline}. 

\vspace{12pt}

\end{document}